\RequirePackage{fix-cm}
\documentclass[onecolumn]{svjour3}          
\smartqed  
\usepackage{graphicx}
\usepackage[normalem]{ulem}

\usepackage{amssymb}
\usepackage{amsmath}
\usepackage{graphicx}
\usepackage{epstopdf}
\usepackage{bm}
\usepackage{epsfig}
\usepackage{mathtools}

\let\vec\bm
\usepackage{color}
\usepackage{algorithm}
\usepackage{algorithmic}

\newcommand{\ignore}[1]{}

\newcommand{\R}{\mathbb{R}}

\begin{document}

\title{Multimodal Information Gain in Bayesian Design of Experiments}


\author{Quan Long}


\institute{30 Blue Ridge Lane, West Hartford, CT, USA \\
              Tel.: +18605018265\\
              \email{quan.spartanlq@gmail.com}         
}

\date{Received: date / Accepted: date}

\maketitle

\begin{abstract}
One of the well-known challenges in optimal experimental design is how to efficiently estimate the nested integrations of the expected information gain. The Gaussian approximation and associated importance sampling have been shown to be effective at reducing the numerical costs. However, they may fail due to the non-negligible biases and the numerical instabilities. A new approach is developed to compute the expected information gain, when the posterior distribution is multimodal - a situation previously ignored by the methods aiming at accelerating the nested numerical integrations. Specifically,  the posterior distribution is approximated using a mixture distribution constructed by multiple runs of global search for the modes and weighted local Laplace approximations. Under any given probability of capturing all the modes, we provide an estimation of the number of runs of searches, which is dimension independent. It is shown that the novel global-local multimodal approach can be significantly more accurate and more efficient than the other existing approaches, especially when the number of modes is large. The methods can be applied to the designs of experiments with both calibrated and uncalibrated observation noises.

\end{abstract}

\section{Introduction}
\label{intro}
Statistical optimal experimental design is a systematic approach to increase data efficiency (i.e., actively acquiring as less data as possible to train a model for given tolerance). In linear optimal experimental design, the optimalities are the norms of the predictive covariance matrix, for example, A-optimality, D-optimality and E-optimality \cite{atkinson2007optimum}\cite{atkinson1992optimum}. A-optimality denotes the trace of the predictive covariance matrix, D-optimality denotes the determinant of the predictive covariance matrix and E-optimality refers to the matrix norm of the predictive covariance matrix. The optimalities can also be correspondingly interpreted from geometrical perspectives. For example, \cite{Titterington} studied the D-optimality from such a point of view.

Nonlinear optimal experimental design is more complicated in that the corresponding optimalities rely on the values of the unknown parameters. The Bayesian approach provides a powerful framework for nonlinear designs in that the prior knowledge of the unknown parameters is naturally incorporated in terms of the relationship between the posterior distribution, the likelihood function and the prior distribution. In a Bayesian framework, the optimal design of an experiment can be obtained by maximizing the expected information gain \cite{Chaloner1995}\cite{Verdinelli2000}. Specifically, the expected information gain is the expected logarithmic ratio of the posterior probability density function (pdf) and the prior pdf of the unknown random parameters \cite{kullback1951}. The expected information gain has been widely used as a measurement of information in engineering and sciences, for example, seismology \cite{LONG2015123}, combustion\cite{BISETTI20160112}, tomography\cite{LONG201324}, biology\cite{dehideniya2019synthetic}, nonlinear dynamical systems \cite{Giovanni}\cite{HUAN2013288} and active learning \cite{Visser2008}.

Obtaining an accurate estimation of the expected information gain is challenging. Using a direct nested sample average \cite{HUAN2013288}\cite{Ryan2003} leads to a computational complexity of ${\cal{O}}(TOL^{-3})$, where $TOL$ is the tolerance on the estimation error \cite{BECK2018523}. Many efforts have been made to accelerate the computations involved in estimating the expected information gain \cite{BECK2018523}\cite{BISETTI20160112}\cite{LONG2015123}\cite{LONG201324}\cite{LONG2015849}, where the Laplace approximation \cite{Tierney1986}\cite{Tierney1989} has been used to analytically complete the inner integration and to provide a proposal distribution in an importance sampling strategy. The Laplace approximation \cite{Tierney1986}\cite{Tierney1989} expands a posterior pdf at the mode and results in a Gaussian integral which approximates the original integration. The discrepancy between the Laplace approximation and the true integral diminishes as the number of data increases \cite{LONG201324}\cite{Tierney1989}. It has also been extended to the cases, where the posterior pdf concentrates on a submanifold in \cite{LONG2015849}. The Laplace-based importance sampling has been used in the sequential design of experiments in \cite{senarathne2019laplace}. It is also note-worthy that the multilevel Monte Carlo method and the Laplace-based importance sampling can be combined to balance the number of samples in the inner and outer integrals, leading to efficient computations of the expected information gain \cite{Beck2019}\cite{Takashi2020}.  The Laplace approximation has been used together with the gradient descent method to tackle the continuous optimization problems in optimal Bayesian experimental design in \cite{carlon2020nesterov}\cite{CHEN2019163}.  Some authors used surrogate to approximate the expected information gain against the design space \cite{Overstall2017}\cite{TARAKANOV2020}. Additionally, it is note-worthy that consistent formulas have been derived for Bayesian optimal experimental design based on infinite dimensional models, for example, the partial differential equations (pdes) \cite{Alexanderian2014} \cite{Alexanderian2018}.

In the previous literature, the applicability of the Laplace approximation relies on the assumption that a single mode dominates the overall shape of the posterior pdf, with the exception of \cite{LONG2015849} which developed the Laplace approximations on the data-informed submanifold. However, this assumption clearly fails to be true, when the posterior pdf is multimodal (see \cite{Ihler2005}\cite{Lan2014}\cite{Shaw2007} for examples of multimodal posterior pdfs). Hence, the imperceptive applications of the Laplace approximation and the importance sampling scheme based on the Laplace approximation would induce irreducible bias and numerical instability.
 
In this paper, a multimodal Laplace approximation is proposed to remove the forementioned restrictions of the Laplace approximation. Under some mild regularity constraints on the posterior pdf, the approach offers the following advantages: First, it removes the bias of the conventional Laplace approximation; second, it is numerically more stable, when used as a proposal distribution in an importance sampling scheme; third, it avoids numerical under-flow \cite{BECK2018523}, when the direct nested sampling is used in a big data scenario.

The outline of the rest of the contents is as follows. Section \ref{sec:LR} briefly reviews Bayesian optimal experimental design, the expected information gain and its numerical approximations. Section \ref{sec:MLA} derives and analyzes the multimodal Laplace approximation. Given the probability of capturing all the modes, a lower bound is provided for the number of independent mode searches. Section \ref{sec:MIS} introduces an importance sampling scheme based on the multimodal Laplace approximations presented in Section \ref{sec:MLA}. The accuracy and efficiency of the proposed methods are demonstrated using two numerical experiments in Section \ref{sec:NE}. 
 
\section{Bayesian optimal experimental design and  expected information gain}\label{sec:LR}
We firstly provide the context of our experimentation. In the current scenario, we can take $m$ observations of the responses of an experiment under the fixed experimental setup. Specifically, we assume the experimental observations can be decomposed into a deterministic model and a measurement noise: 
\[\vec y_i = \vec g(\vec\theta_t, \vec\xi) + \vec \epsilon_i \quad\text{with} \quad \vec \epsilon_i \sim {\cal{N}}(\vec 0, \vec\Sigma_e) \quad \text{and}\quad i=1,...,m\,,\] 
where $\vec y_i \in \R^{d_m} $ is the $i^{th}$ observation of the measurement data, $\vec \theta_t \in {\R}^d$ is the vector of the ``true'' parameters, $\vec g(\vec\theta_t, \vec\xi)$ is a deterministic function of the system parameters and the experimental setup, $\vec\epsilon_i$ is the observational noise, $\vec\Sigma_e$ is the covariance matrix of the random measurement noise,  $\vec\xi \in \R^s$ is the vector of the design parameters, for example, the locations of the electrodes in a tomography experiment \cite{LONG201324} and the initial temperature in a combustion experiment \cite{BISETTI20160112}. For the sake of conciseness, we denote $\bar{\vec y} = [\vec y_1\,, \vec y_2\,, ..., \vec y_m]^{\top}$ in the rest of the paper. Note that we restrict our methodology to batch experimental design and multiple different experimental setup can be considered by the multi-dimensional $\vec\xi$ vector.

 Consequently, the Bayes Theorem can be expressed as follows:
\begin{align}
p(\vec\theta | \bar{\vec y}, \vec\xi) = \frac{p(\bar{\vec y} | \vec \theta, \vec \xi)p(\vec \theta)}{p(\bar{\vec y}|\vec\xi)}\,. \label{eq:BayesAllExp}
\end{align}
where $p(\cdot)$ denotes a pdf. $p(\bar{\vec y} | \vec \theta, \vec\xi)$, $p(\vec \theta)$, $p(\bar{\vec y}|\vec\xi)$ and $p(\vec\theta | \bar{\vec y}, \vec \xi)$ denote the likelihood function, the prior pdf, the marginal likelihood function (also called the evidence), and the posterior pdf, respectively.

The current derivation of Laplace approximation is based on the additive model of data and the noise is restricted to be Gaussian. That said, the validity of Laplace approximation is not compromised, while other scenarios are considered, for example, data $\vec y$ are discrete random variables, noises in the responses are stemmed from intrinsic randomness \cite{Calvetti}.

Note that we can similarly derive a Bayesian framework for the parameter identification of time dependent problems, which has been demonstrated in \cite{LONG2015123}. Additionally, the assumption of additive noise is a flexible one, for example, a multiplicative model (i.e., $\vec y = \vec g\cdot \vec\epsilon$) can be converted to its additive form via a logarithmic transformation of the involved variables. 

We focus on the information efficiency of the experiments denoted by the setup $\vec \xi$, which controls the generation of data $\bar{\vec y}$ in terms of model $\vec g$. A widely used metric is the so-called expected information gain \cite{kullback1951}\cite{Chaloner1995}. It measures the amount of new ``information'' encoded in $p(\vec\theta|\bar{\vec y}, \vec \xi)$ on top of the legacy ``information'' represented by $p(\vec\theta)$. It is essentially an averaged Kullback-Leibler (K-L) divergence \cite{kullback1951} (also called discrimination information \cite{Ghosh1987}) against all the possible values of the data. It can be used as a utility function to quantify the information contents in the data with respect to (w.r.t.) the unknown parameters and other focused unknown quantities of interest (QoI) \cite{LONG201324}. The expected information gain associated to \eqref{eq:BayesAllExp} can be written as the following:
\begin{align}
I = \int_{\cal{\vec Y}}\int_{{\vec\Theta}} log\left[\frac{p(\vec\theta | \bar{\vec y})}{p(\vec\theta)}\right]p(\vec\theta | \bar{\vec y}) p(\bar{\vec y}) d\vec\theta d\bar{\vec y}\,,\label{eq:KL} 
\end{align}
where we ignore the notation of $\vec\xi$ in the conditional distributions for the sake of conciseness. To the best of our knowledge, \cite{Ryan2003} is the first paper which noted the computational challenge of estimating the utility function using the so called double loop Monte Carlo method (DLMC). The DLMC sampler, which has been benchmarked against by many more efficient methods in recent years, is a direct nested random sampling method: 
\begin{align} 
I_{DLMC}=\frac{1}{M_1} \sum^{M_1}_{i=1}log\left[\frac{p(\bar{\vec y}_i | \vec \theta_i)}{ \frac{1}{N_1} \sum^{N_1}_{j=1} p(\bar{\vec y}_i | \vec \theta_j)}\right] \,,\label{eq:DLMC} 
\end{align}
where $M_1$ is the number of samples in the outer loop, and $N_1$ is the number of samples in the inner loop. 
$\vec\theta_i\sim p(\vec\theta)$, $\bar{\vec y}_i\sim p(\bar{\vec y} | \vec \theta_i)$, and $\vec\theta_j\sim p(\vec\theta)$.
The variance and bias of this estimator can be bounded by $\mathop{\mathbb{V}}(I_{DLMC}) = {\cal{O}}\left(\frac{1}{M_1}\right)$ and $\mathop{\mathbb{B}}(I_{DLMC})={\cal{O}}\left(\frac{1}{N_1}\right)$, respectively (see \cite{Ryan2015}\cite{LONG201324}\cite{LONG2015849} for details of how the errors are estimated ). The computational cost of this estimator can be characterized by the number of likelihood computations (i.e., $W_1 = N_1\times M_1$). 

If a numerical tolerance is imposed on the mean square error: 
\[\mathop{\mathbb{V}}(I_{DLMC}) + \mathop{\mathbb{B}}(I_{DLMC})^2 = TOL^2\,,\]
we can obtain the magnitudes of $N_1$ and $M_1$ in terms of the tolerance: 
\[ M_1 = {\cal{O}}\left( TOL^{-2}\right)\quad \text{and} \quad N_1 = {\cal{O}}\left( TOL^{-1}\right)\,.\]
Consequently, the cost of the total computational cost is proportional to the cubic of the reciprocal of the tolerance:
\[W_1 = {\cal{O}}\left(TOL^{-3}\right)\,.\]

Note that the bias caused by the inexact numerical approximations of the forward model $\vec g(\vec\theta)$ is not considered in this study. Interested readers can refer to \cite{BECK2018523} from this perspective.  


One of the more efficient approaches than the DLMC method is to invoke the Laplace approximation of the inner integral. The Laplace approximation \cite{laplace} (also referred to as Laplace method) is conventionally used to convert the integral of an exponential function to a Gaussian integral, which can be written as follows:
\begin{align}
\int e^{-mf(x)}dx = \sqrt{\frac{2\pi}{m|f''(\hat{x})|}}e^{-mf(\hat{x})} + {\cal{O}}\left(\frac{1}{m}\right),\, \label{eq:laplace}
\end{align}
where $\hat{x}$ is the global optimum of $f(x)$. The asymptotic error, ${\cal{O}}\left(\frac{1}{m}\right)$, may have different rate w.r.t. $m$, dependent on how the expansion of $f(x)$ is truncated at $\hat{x}$  \cite{Tierney1986}\cite{Tierney1989}\cite{schillings2020}. 
Expanding the log-posterior function at the maximum a posterior (MAP) estimate and following a multivariate version of \eqref{eq:laplace} (see \cite{LONG201324} for details of the derivation procedure and error analysis),
the K-L divergence (the inner integral of \eqref{eq:KL} ) can be analytically obtained using the Laplace approximation. Consequently, the expected information gain can be approximated as follows: 
\begin{align}
I =&\int_{\Theta}\int_{\cal{\vec Y}} \left[-\frac{1}{2}log\left( |\vec\Sigma| \right) - \frac{\vec{tr}(\vec\Sigma \vec H_h(\hat{\vec\theta}))}{2}\right]p(\bar{\vec y}|\vec\theta_t)p(\vec\theta_t)d\bar{\vec y}d\vec\theta_t\nonumber \\&- \frac{d}{2} - \frac{d}{2}log(2\pi) + {\cal{O}}\left(\frac{1}{m}\right)\,,\label{eq:LA1}
\end{align}
where the posterior covariance matrix can be approximated by the inverse of the Hessian matrix of the log-posterior function at the MAP estimate (i.e., $\vec\Sigma \approx \vec H^{-1}(\hat{\vec\theta}) \approx \left[\nabla_{\vec\theta}\vec g^{\top}(\hat{\vec\theta})\vec\Sigma^{-1}_{\epsilon}\nabla_{\vec\theta}\vec g(\hat{\vec\theta})\right]^{-1}$), $\hat{\vec\theta}$ is the MAP estimate of $\vec\theta$ conditioned on the synthetic data $\bar{\vec y}$, $\vec{tr}(\cdot)$ denotes matrix trace and $\vec{H}_h(\hat{\vec\theta})= \nabla_{\vec\theta}\nabla_{\vec\theta} log p(\hat{\vec\theta})$.  The corresponding numerical discritization of \eqref{eq:LA1} can be written as follows:
\begin{align}
I_{LA} =& \frac{1}{M_2} \sum^{M_2}_{i=1} \left[-\frac{1}{2}log\left( |\vec\Sigma_i| \right) - \frac{\vec{tr}(\vec\Sigma_i \vec H_h(\hat{\vec\theta}_i))}{2}\right]
- \frac{d}{2} - \frac{d}{2}log(2\pi)\,,\label{eq:LA2}
\end{align}
where the $\hat{\vec\theta}_i$ is the MAP estimate of $\vec\theta$ conditioned on the synthetic data $\bar{\vec y}_i \sim p(\bar{\vec y}|{\vec{\theta}_t}^i)$ and ${\vec{\theta}_t}^i \sim p(\vec\theta_t)$.
The computational cost of estimating the inner integral of \eqref{eq:DLMC} by running thousands times the likelihood functions can now be reduced to one run of optimization to search for $\hat{\vec\theta}_i$. 

In the scenario where the posterior Hessian matrix $\vec H(\hat{\vec\theta})$ is low-rank \cite{LONG2015849} (i.e., $r < d$ with $r$ be the rank of $\vec H(\hat {\vec\theta})$), new information is only gained in the subspace of the parameters, which are sensitive to the perturbations of the observed data. The non-intuitive relationship between the expected information gain and the nonzero eigenvalues of the posterior Hessian matrix has been revealed in \cite{LONG2015849}. The Laplace approximation of the expected information gain can be extended to the following form:

\begin{align}
I=\int_{\Theta}\int_{{\cal{Y}}} \left[ -\frac{1}{2}log |\vec\Sigma_p| 
 - \frac{r}{2} - \frac{r}{2}log(2\pi)-log\int_{\vec{T}} p_{\vec s,\vec t}(\vec 0,\vec t)d\vec t \right] p(\bar{\vec y}|\vec\theta_t)p(\vec\theta_t)d\vec\theta_t+ {\cal{O}}\left(\frac{1}{m}\right) \,, \label{eq:LAM1}
\end{align}
where $\vec\Sigma_p(\hat{\vec\theta}) = \left[\vec U^{\top}(\hat{\vec\theta})\vec H(\hat{\vec\theta}) \vec U(\hat{\vec\theta})\right]^{-1}$, the columns of $\vec U$ are the basis spanning the orthogonal space of the Jacobian kernal of $\vec H$, $\int_{\vec{T}} p_{\vec s,\vec t}(\vec 0,\vec t)d\vec t $ is an integral on the non-informative manifold $\vec T$, $\vec t$ is the new variable which parameterizes $\vec T$ and $\vec s$ is the variable which parameterizes the normal direction of the manifold, $p_{\vec s,\vec t}(\vec s,\vec t)$ is the prior pdf after a change of parameters (i.e., $p_{\vec s,\vec t}(\vec s,\vec t)d \vec s d\vec t = p(\vec\theta) d \vec\theta  $).

The equation \eqref{eq:LAM1} can be discritized as 
\begin{align}
I_{LAS}= \frac{1}{M_3} \sum^{M_3}_{i=1} \left[ -\frac{1}{2}log |\vec\Sigma_p^i| \right]
 - \frac{r}{2} - \frac{r}{2}log(2\pi)-log\int_{\vec{T}} p_{\vec s,\vec t}(\vec 0,\vec t)d\vec t \,,\label{eq:LAM2}
\end{align}
where $\vec\Sigma_p^i(\hat{\vec\theta}) = \left[\vec U^{\top}(\hat{\vec\theta}_i)\vec H(\hat{\vec\theta}_i) \vec U(\hat{\vec\theta}_i)\right]^{-1}$. Note that \eqref{eq:LAM1} and \eqref{eq:LAM2} become  \eqref{eq:LA1} and \eqref{eq:LA2}, respectively, when $r = d$. 

Although Laplace approximation enables efficient estimation of the inner loop of \eqref{eq:DLMC}, it entails an asymptotic bias when the posterior distribution is non-Gaussian. To achieve both acceleration and consistency, importance sampling based on Laplace approximation has been developed in \cite{BECK2018523} and \cite{Ryan2015} for optimal experimental design. The specific scheme can be written as follows

\begin{align}
I_{LAIS} =\frac{1}{M_4} \sum^{M_4}_{i=1}log\left[\frac{p(\bar{\vec y}_i | \vec \theta_i)}{ \frac{1}{N_4} \sum^{N_4}_{j=1} p(\bar{\vec y}_i | \vec \theta_j)\alpha_j^i}\right]\,,\label{eq:LAIS}
\end{align}
where the likelihood ratio, $\alpha_j^i = \frac{p(\vec\theta_j)}{p_g^i(\vec\theta_j)}$, $p_g^i$ is a Gaussian pdf with mean $\hat{\vec\theta}_i$ and covariance matrix 
$\vec\Sigma_i$, $\vec\theta_j \sim {\cal{N}}(\hat{\vec\theta}_i, \vec\Sigma_i)$. On top of being efficient, direct Laplace approximation and corresponding importance sampling can both avoid the numerical underflow as demonstrated in \cite{BECK2018523}, which occurs with very high probability in DLMC, when a large amount of data lead to a very concentrated posterior pdf. Note that a Student's t-distribution should be used as the proposal distribution in the importance sampling scheme, if the  posterior pdf has a heavy tail. 

\begin{table}
\begin{center}
    \begin{tabular}{| c | c |}
    \hline
        Method & Complexity  \\ \hline
        DLMC & $W_1=M_1\times( C_l + N_1\times C_l)$  \\
\hline
    LA & $W_2=M_2\times C_o $   \\ \hline
    Manifold LA & $W_3=M_3\times(C_o + C_e)$ \\ \hline
        LAIS & $W_4=M_4\times(C_o + C_l + N_4\times C_l)$ \\ \hline
    \end{tabular}
    \caption{Computational costs of the numerical schemes in the literature.}\label{tab1}
\end{center}
\end{table}

The computational costs of the numerical schemes, which we have reviewed so far, are listed in Table \ref{tab1}, where $C_o$ denotes the cost of a single run of optimization algorithm (e.g. the gradient descent method) used to find the MAP estimate and the associated Hessian matrix, $C_l$ is the cost of computing a likelihood function including computing the forward model $\vec g(\vec\theta)$ and the associated pdf, $C_e$ is the cost of the numerical eigenvalue decomposition of the Hessian matrix.

\begin{table}
\begin{center}
    \begin{tabular}{| c | c |}
    \hline
        Number of samples & Estimation  \\ \hline
        $M_1$ & $\mathop{\mathbb{V}}\left( log\left[\frac{p(\bar{\vec y}| \vec \theta)}{ p(\bar{\vec y} )}\right]\right)\times TOL^{-2}$  \\
\hline
        $M_2$ &  $\mathop{\mathbb{V}}\left( \frac{1}{2}log|\vec H(\vec\theta)|\right)\times TOL^{-2}$ \\
\hline
        $M_3$ &  $\mathop{\mathbb{V}}\left( \frac{1}{2}log|\vec H_p(\vec\theta)|\right)\times TOL^{-2}$  \\
\hline
        $M_4$ & $\mathop{\mathbb{V}}\left( log\left[\frac{p(\bar{\vec y}| \vec \theta)}{ p(\bar{\vec y} )}\right]\right)\times TOL^{-2}$  \\
\hline
        $N_1$ & $\frac{1}{2}{\mathop{\mathbb{E}} }\left(\mathop{\mathbb{V}} \left[\frac{p(\bar{\vec y}| \vec \theta)}{ p(\bar{\vec y} )} | \bar{\vec y} \right]\right)\times TOL^{-1} $ \\
\hline
        $N_4$ & $\frac{1}{2}{\mathop{\mathbb{E}} }\left(\mathop{\mathbb{V}} \left[\frac{p(\bar{\vec y}| \vec \theta)\alpha(\vec\theta)}{ p(\bar{\vec y} )} | \bar{\vec y} \right]\right)\times TOL^{-1}$  \\
\hline
    \end{tabular}
    \caption{Estimation of the numbers of samples.}\label{tab2}
\end{center}
\end{table}

We tabulate the estimates of the number of samples in all of the introduced estimators in Table \ref{tab2}. Note that $N_4 \ll N_1 $ for a concentrated single modal posterior pdf, because of the much reduced variance in the importance sampling scheme of LAIS. 

\section{Multimodal Laplace approximation (MLA)}\label{sec:MLA}
The numerical integrations based on the Laplace approximations are only valid in the cases, where the posterior pdf has a single dominant mode. An example of Gaussian mixture is used here to illustrate the potential erroneous information gain caused by a single modal Laplace approximation. In Figure \ref{fig:multimodal}, the solid curve denotes a Gaussian mixture pdf of three modes and the dashed line denotes a Gaussian pdf matching the first mode of the mixture pdf. The multimodal pdf involves three weighted modes: 
$0.3\times {\cal{N}}(2, 0.5) + 0.3\times {\cal{N}}(5, 0.2) 
+ 0.4\times {\cal{N}}(7, 0.5)$. Note that the error in considering a Gaussian pdf rather than a truncated Gaussian pdf is negligible. 

\begin{figure}[ht]
\centering
\includegraphics[scale=0.55]{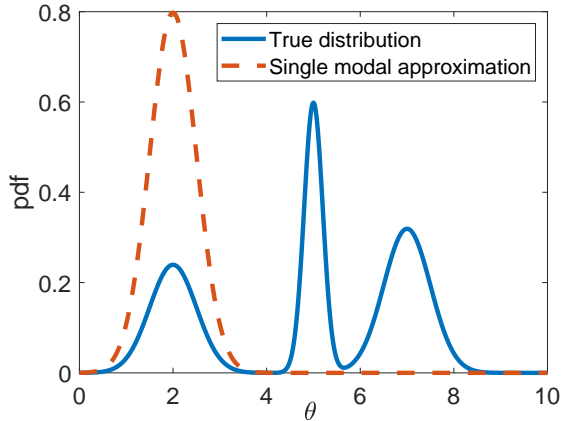}
\caption{A Gaussian mixture pdf and its single modal approximation using the Laplace approximation.}\label{fig:multimodal}
\end{figure}

Its single modal approximation can be the Gaussian distribution: ${\cal{N}}(2, 0.5)$. Using a uniform distribution, ${\cal{U}}(0,10)$, as the prior pdf, and decomposing the K-L divergence into three integrals corresponding to the disparate modes, the approximated information gain can be analytically estimated by the following expression
\begin{align}
D_{KL} \approx \sum^3_{k=1}D_{KL}(w_k{\cal{N}}(\mu_k, \sigma^2_k)\, ||\, {\cal{U}}(0,10)) = \sum^3_{k=1}\left[ w_k log(w_k) - 0.5 w_k log(\sigma^2_k) - w_k h\right]\,, \label{eq:multi_single_error}
\end{align}
where $w_k$, $\mu_k$ and $\sigma^2_k$ are the weight, the mean and the variance of the $k^{th}$ Gaussian mode, respectively, $h = - log 10$ is the logarithm of the prior pdf. Substituting the values of $w_k$, $\mu_k$ and $\sigma^2_k$ into \eqref{eq:multi_single_error}, we obtain ${D_{KL}}\approx 0.97$. To the contrary, a conventional Laplace approximation taking the first mode as the single dominant mode leads to $D_{KL}^{laplace}\approx 0.5log(\sigma_1^2)+  h(\mu_1)\approx 1.23$.

%
%

In a special case, where $w_k = 1/K$ and $\sigma^2_k = \sigma^2$, $k=1,...,K$, the discrepancy between the conventional Laplace approximation and the true value of the information gain can be found as the following:
\[D_{KL} -  D_{KL}^{laplace} = logK \,, \]
which increases as the number of modes increases.
 
Next, we present the multidimensional multimodal approximation of the expected information gain in Theorem 1. 

\begin{theorem}
Assuming that the posterior distribution of $\vec\theta$ conditioning on the synthetic data $\bar{\vec y} $ entails $K$ modes, data being modeled by $\vec y_i = \vec g(\vec\theta_t) + \vec\epsilon_i$, the third derivatives of the model $\vec g(\vec\theta_t)$ being bounded from above, we can approximate the expected information gain as follows
\begin{align}
I=& \int_{\vec\Theta} \int_{{\cal{\vec Y}}} \sum^K_{k=1}[ w_k log(w_k) - 0.5 w_k log(|\vec\Sigma_k|)\nonumber \\
& - w_k h(\hat{\vec\theta}_k)] p(\bar{\vec y} |\vec\theta_t) p(\vec\theta_t) d\bar{\vec y} d\vec\theta_t  -0.5log(2\pi)^d \nonumber\\
&- 0.5d +{\cal{O}}\left( \frac{1}{m} \right)\,, \label{eq:theorem1}
\end{align}
where $\hat{\vec\theta}_k$ is the $k^{th}$ mode of the posterior pdf, $w_k$ 
is the weight of the $k^{th}$ mode, $\vec\Sigma_k$ is the inverse of the Hessian matrix of the negative logarithm of the posterior pdf at the $k^{th}$ mode and $h(\hat{\vec\theta}_k)=log(p(\hat{\vec\theta}_k))$. 
\end{theorem}

\begin{proof}
We first express the integral of the K-L divergence as a summation of several integrals on the sub-domains related to the multiple disparate modes:
\begin{align}
D(\bar{\vec y}) =& \sum^K_{k=1} \int_{\Omega_k} log\left[\frac{p(\vec\theta | \bar{\vec y})}{p(\vec\theta)}\right]p(\vec\theta | \bar{\vec y}) d\vec\theta + {\cal{O}}\left( e^{-m} \right)\,,\label{eq:KL2}
\end{align}
where $\Omega_k$ is a ball of radius $R_k$ covering a neighborhood of the $k^{th}$ mode. According to the Laplace principle \cite{Dembo}, most of the probability mass is concentrated at the modes as $m$ increases, namely, a big{\color{red}{-}}data scenario \cite{schillings2020}. The integral outside $\vec\Omega = {\displaystyle\cup}_{k=1}^K\Omega_k$ adds up to 
 $  Ce^{-m}$ with the constant $C = e^ {ess\,inf_{\vec\theta \in \vec\Theta \setminus \Omega}\phi(\vec\theta)} $, $\phi(\vec\theta)$ is the integrand function in \eqref{eq:KL2} scaled by the number of data:
\[\phi(\vec\theta) = \frac{1}{m} log\left[\frac{p(\vec\theta | \bar{\vec y})}{p(\vec\theta)} \right] p(\vec\theta | \bar{\vec y})\,.\]
 
 Next, locally approximating $log[p(\vec\theta | \bar{\vec y})]$ in each $\Omega_k$ by a truncated second order Taylor expansion leads to:
 \begin{align}
 p(\vec\theta | \bar{\vec y}) = \tilde{p}_k(\vec\theta | \bar{\vec y})  + 
 e^{\left[{\cal{O}}\left(|\vec\theta - \hat{\vec\theta}_k |^3  \right)   \right]}\,,\nonumber
 \end{align}
 where
 
\begin{align} 
\tilde{p}_k(\vec\theta | \bar{\vec y})&= p(\hat{\vec\theta}_k| \bar{\vec y} ) exp \left[ - \frac{1}{2} (\vec\theta - \hat{\vec\theta}_k )^{\top} \vec \Sigma_k^{-1} (\vec\theta - \hat{\vec\theta}_k )\right]\nonumber\\
&= w_k p^g_k (\vec\theta)\,,\nonumber
\end{align}
 with $p^g_k (\vec\theta)$ be the pdf of the Gaussian distribution-
${\cal{N}}(\hat{\vec\theta}_k, \vec\Sigma_k)$, with $ \vec{\Sigma}_k^{-1} = -\nabla_{\vec\theta}\nabla_{\vec\theta} log \left[p(\hat{\vec\theta}_k| \bar{\vec y})\right]$ is the Hessian matrix of the negative logarithm of the posterior pdf at $\hat{\vec\theta}_k$ and $w_k=p(\hat{\vec\theta}_k | \bar{\vec y})\sqrt{2\pi}^d |\vec{\Sigma}_k|^{1/2}$. Note that the first derivative of $log [p(\vec\theta | \bar{\vec y})]$ at $\hat{\vec\theta}_k$ is zero (i.e., $\nabla_{\vec\theta} log\left[p(\hat{\vec\theta}_k | \bar{\vec y})\right]= 0$). Substituting $p(\vec\theta|\bar{\vec y})$ by $\tilde{p}_k(\vec\theta|\bar{\vec y})$ in $\Omega_k$,  we can approximate the K-L divergence of \eqref{eq:KL2} as follows:

\begin{align}
D(\bar{\vec y}) =& 
\sum^K_{k=1} \int_{\Omega_k} log\left[\frac{\tilde{p}_k(\vec\theta | \bar{\vec y})}{p(\vec\theta)}\right]\tilde{p}_k(\vec\theta | \bar{\vec y}) d\vec\theta +{\cal{O}}\left( \frac{1}{m} \right) \nonumber \\
=&\sum^K_{k=1} \int_{\Omega_k} log\left[\frac{w_k p^g_k(\vec\theta )}{p(\vec\theta)} \right]  w_k p^g_k(\vec\theta) d\vec\theta + {\cal{O}}\left( \frac{1}{m} \right)\nonumber\\
=&\sum^K_{k=1} w_k log( w_k) + \sum^K_{k=1} \int_{\Omega_k} w_k log\left[\frac{p^g_k(\vec\theta )}{p(\vec\theta)}\right]p^g_k(\vec\theta) d\vec\theta +{\cal{O}}\left( \frac{1}{m} \right)\,. \label{eq:tildeD}
\end{align}

The second term in the last line of \eqref{eq:tildeD} can be obtained analytically by expanding the prior pdf locally around $\hat{\vec\theta}_k$, and the Laplace approximation of a K-L divergence, in the case of a multidimensional multimodal posterior pdf, can be expressed as
\begin{align}
D(\bar{\vec y} ) = \tilde{D}(\bar{\vec y})
+ {\cal{O}}\left( \frac{1}{m} \right) \nonumber
\end{align}
with
\begin{align}
\tilde{D}(\bar{\vec y})=&
\sum^K_{k=1} \left[ w_k log (w_k)  - 0.5 w_k log\left[|\vec\Sigma_k(\hat{\vec\theta}_k)|\right] - w_k h(\hat{\vec\theta}_k) \right]\nonumber\\
& - 0.5 log(2\pi)^d - 0.5 d\,. \label{eq:multiD_KL}
\end{align}

Consequently, the corresponding expected information gain can be expressed as 
\begin{align}
I = &\int_{{\cal{\vec Y}} } D(\bar{\vec y}) p(\bar{\vec y}) \bar{\vec y} =\int_{\vec\Theta}\int_{{\cal{\vec Y}}  } D(\bar{\vec y}) p(\bar{\vec y}|\vec\theta_t)p(\vec\theta_t)d \bar{\vec y} d\vec\theta_t\nonumber \\
=&\int_{\vec\Theta}\int_{{\cal{\vec Y}}  } \tilde{D}(\bar{\vec y}) p(\bar{\vec y}|\vec\theta_t)p(\vec\theta_t)d \bar{\vec y} d\vec\theta_t + {\cal{O}}\left( \frac{1}{m} \right)
\end{align}  
$\square$
\end{proof}

The posterior modes are commonly unknown and need to be obtained numerically. We find the modes using runs of optimization starting from {\color{red}{a}} few randomized initial points. Algorithm 1 shows the procedure of computing the expected information gain via multimodal Laplace approximation. 

\begin{algorithm}[pht]
\caption{Multimodal Laplace approximation}\label{alg:MLA}
\begin{algorithmic} [1]
\STATE inputs: $p(\vec\theta)$, $g(\vec\theta)$, $\vec \xi$, $n$ 
\STATE draw a sample from the prior distribution 
$\vec\theta^i_t \sim p(\vec\theta)$   
\STATE draw a sample of data from the likelihood function $\bar{\vec y}_i \sim p(\bar{\vec y} | \vec\theta^i_t)$   
\STATE solve the minimization problem: $\arg\min_{\vec\theta} [-log(p(\vec\theta|\bar{\vec y}_i   ))]$ $n$ times with distinct initial points, the distinct local optimal solutions are $\{\hat{\vec\theta}^i_k, k=1,...,K\}$ (note that $K$ can also vary w.r.t. $i$)
\STATE compute $\tilde{D}_i$ of \eqref{eq:multiD_KL} using the Hessian (${\vec\Sigma^i_k}^{-1}$) and the normalized likelihood value ($w_k^i$) 
\STATE  Repeat steps 2-5 $M_5$ times and take the sample average of $\tilde{D}_i$: $I_{MLA} = \frac{1}{M_5} \sum_{i=1}^{M_5} \tilde{D}_i$ 
\end{algorithmic}
\end{algorithm}

The cost of computing the expected information gain using Algorithm 1 can be estimated as the following: 
\[W_5 = M_5 \times( n \times C_o + C_l)\,, \]
where $M_5$ is the number of samples of $\vec\theta_t$ drawn from the prior distribution, $n$ is the number of optimization runs. 

In addition, $n$ can be estimated for any given probability that all modes are captured using Algorithm 1. Such probability can be expressed as the following:
\begin{align}
1-P(\cup e_i) &= 1 - \left[\sum_{i=1}^K P(e_i) - \sum_{i>j} P(e_i \cap e_j) + \sum_{i>j>k} P(e_i \cap e_j \cap e_k) - ...\right] \nonumber\\
& \geq 1 - \sum_i^K P(e_i)\,, \nonumber
\end{align}
where $Pr(\cdot)$ is a probability function, $Pr(\cup e_i)$ is the probability that at least one mode is missed by the $n$ runs of optimization, $e_i$ is the set of events where the $i^{th}$ mode is missed in all the $n$ runs of optimization. 
The probability of $e_i$ can be written as 
\[Pr(e_i) = (1-p_i)^n\], 
where $p_i$ is the probability that an i.i.d. initial value of $\theta$ converges to the $i^{th}$ mode. It is straightforward to see that the following relationship exists: 
\begin{align}
1 - \sum_i^K Pr(e_i) \leq 1 -  K (1-p)^n \,, \label{eq:ei}
\end{align}
where $p = min(p_i)$. Let $\beta = \sum_i^K Pr(e_i) $, we can obtain the following inequality from \eqref{eq:ei}:
\[\beta \geq K(1-p)^n\,.\] 
Subsequently, we have the following lower bound of $n$:
\begin{align}
n \geq \frac{log\beta - log K}{log(1-p)}\,,
\label{eq:nor}
\end{align} 
which is independent to the dimension of $\vec\theta$. Note that $p$ is associated to the smallest size of the basins of attractions of all the modes.
\section{Multimodal nested importance sampling (MNIS)}\label{sec:MIS}
The method in the previous section approximates a multimodal posterior pdf using a mixture Gaussian pdf and replaces the inner integration in \eqref{eq:KL} using multiple weighted Gaussian integrals. It provides a remedy to the problems caused a blindly applied conventional Laplace approximation. However, its error term, ${\cal{O}}\left( \frac{1}{m} \right)$, may not be negligible, when $m$ is not so large.

Importance sampling is an ideal method to marry the powers of the Laplace approximation and random sampling. The Laplace approximation based importance sampling has been demonstrated with salient efficiency and accuracy, when it is used to calculate a posterior expectation, for example, the expected information gains used in Bayesian optimal experimental design  \cite{BECK2018523}\cite{Ryan2015}\cite{schillings2020}. A multimodal importance sampling method based on the multimodal Laplace approximations is presented in this section to relieve the numerical instability (i.e., value of the likelihood ratio be close to infinity, when samples are taken from the tail of the proposal pdf) induced by inappropriately using a Gaussian proposal. Theorem 2 summarizes the results of the change of measure. 
 
\begin{theorem}\label{th2}
Assuming that the posterior distribution of $\vec\theta$, conditioning on the synthetic data $\bar{\vec y} $, entails $K$ modes, data being modeled by $\vec y_i = \vec g(\vec\theta_t) + \vec\epsilon_i$, the third derivatives of the model $\vec g(\vec\theta_t)$ being bounded from above, for any confidence level $\alpha>0$, there exist $M_0$ and $N_0$, for all $M_6>M_0$ and $N_6>N_0$, we have 
\begin{align}
&Pr\left(\left| I-\frac{1}{M_6}\sum^{M_6}_{i=1}log\left[\frac{p(\bar{\vec y}_i | \vec \theta_i)}{ \frac{1}{N_6} \sum^{N_6}_{j=1} p(\bar{\vec y}_i | \vec \theta_j)\beta^i_j}\right]\right| > \frac{C_1}{\sqrt{M_6}}  + \frac{C_2}{N_6}\right) < \alpha \,,\label{eq:MNIS}
\end{align}
with 
\begin{align}
C_1= \sqrt{\mathop{\mathbb{V}}\left( log\left[\frac{p(\bar{\vec y}| \vec \theta)}{ p(\bar{\vec y} )}\right]\right)}\nonumber
\,,\, 
\end{align}
and
\begin{align}
C_2 = \frac{1}{2}{\mathop{\mathbb{E}} }\left(\mathop{\mathbb{V}} \left[\frac{p(\bar{\vec y}| \vec \theta)\beta(\vec\theta)}{ p(\bar{\vec y} )} | \bar{\vec y} \right]\right)\,,\nonumber
\end{align}
where $Pr(\cdot)$ denotes a probability function, $\beta^i_j = \frac{p(\vec\theta_j)}{p^i_{gm}(\vec\theta_j)}$, $p^i_{gm}$ is the pdf of the Gaussian mixture $\sum^K_{k=1}m_k {\cal{N}}\left(\hat{\vec\theta}^i_k, \vec\Sigma^i_k \right)$, $M_6$ is the number of samples of the unknown parameters drawn from the prior distribution, $N_6$ is the number of samples of the unknown parameters drawn from the Gaussian mixture. 
\end{theorem}

\begin{proof}
The theorem is a direct entailment of Proposition 1 in \cite{BECK2018523}. 
$\square$
\end{proof}

Algorithm \ref{alg:MMIS} lists the steps of computing the expected information gain via a change of Gaussian mixture measure, with the errors stated in Theorem \ref{th2}.

\begin{algorithm}[pht]
\caption{Importance sampling based on multimodal Laplace approximation}\label{alg:MMIS}
\begin{algorithmic} [1]
\STATE inputs: $p(\vec\theta)$, $g(\vec\theta)$, $\vec \xi$, $n$ 
\STATE draw a sample from the prior distribution 
$\vec\theta^i_t \sim p(\vec\theta)$   
\STATE draw a sample of data from the likelihood function $\bar{\vec y}_i \sim p(\bar{\vec y} | \vec\theta^i_t)$   
\STATE solve problem $\arg\min_{\vec\theta} [-log(p(\vec\theta|\bar{\vec y}_i  ))]$ $n$ times with distinct initial points, the distinct local optimal solutions are $\{\hat{\vec\theta}^i_k, k=1,...,K\}$ (note that $K$ can also vary w.r.t. $i$)
\STATE  draw $N_6$ i.i.d. samples from Gaussian mixture: $\vec\theta_j \sim \sum^K_{k=1} w_k {\cal{N}}( \hat{\vec\theta}^i_k, \vec\Sigma^i_k )$
\STATE  compute the inner sample average in \eqref{eq:MNIS}: $L_i=\frac{1}{N_6} \sum^{N_6}_{j=1} p(\bar{\vec y}_i | \vec \theta_j)\beta^i_j$ 
\STATE Repeat steps 2-6 $M_6$ times and compute the outer sample average in \eqref{eq:MNIS}:
$I_{MNIS}= \frac{1}{M_6}\sum^{M_6}_{i=1}log\left[\frac{p(\bar{\vec y}_i | \vec \theta_i)}{ L_i}\right]$ 
\end{algorithmic}
\end{algorithm}

The computational complexity of Algorithm 2 can be expressed as follows:
\[W_6 = M_6 \times [n\times C_o + C_l + N_6\times C_l ]\,.\]
Thanks to the change of measure, $C_2$ in Theorem 2 is very small and only a small $N_6=C_2/(\gamma TOL)$  is needed to control the bias of this multimodal importance sampler. $\gamma TOL$ is the tolerance on the bias. Consequently, the following inequality can be true in practice 
\[N_6\times C_l \ll n\times C_o + C_l \,. \]
Therefore, the computational cost of the MNIS method is similar to the cost of the direct multimodal Laplace approximation. We will demonstrate the efficacy of the proposed methods using numerical examples in the next section. 
 
\section{Numerical validation}\label{sec:NE}

\subsection{A simple model of multiple factors}
First, we prove the concept using a multidimensional model consisting of quadratic monomials:

\begin{align}
\vec y = \vec g + \vec \epsilon \quad \text{with} \quad \vec g = [\xi\theta_1^2 \quad (1-0.5\times\xi)\theta_2^2 \quad \theta_3^2]^{\top}\,,
\end{align}
where $\vec y \in \vec R^3$ is the vector of the synthetic data, $\vec g$ is a three-dimensional deterministic model, whose components consist of distinct monomials in each dimension, $\vec\epsilon \sim {\cal{N}}\left(\vec 0, \vec \Sigma_e \right)$ is a three-dimensional Gaussian random vector with mean zero and covariance matrix $\vec\Sigma_e=\vec I \sigma^2_e$. The unknown parameters are uniformly distributed between $-10$ and $10$ (i.e., $\theta_i \sim {\cal{U}}(-10, 10),\, i=1,2,3$). 

Figure \ref{fig:convergence} demonstrates the superior performance of the MLA and MNIS methods against the DLMC method (the baseline). The bias of the MLA method is negligible even with only $10$ sampling points (i.e., $M_5=10$ in Algorithm 1). To the contrary, the error of DLMC method is dominated by a bias, which is controlled by the number of samples in the inner loop (i.e., $N_1$ in equation \eqref{eq:DLMC}). It is shown that at least $10^4$ inner samples are needed to reduce the relative error below $10\%$. Two curves of the MLA method produced by $10$ and $20$ distinct optimization runs are plotted on the left of Figure \ref{fig:convergence}. 

An estimation of the required optimization runs using \eqref{eq:nor} is 
$\frac{log(0.1)-log(8)}{log(7/8)}\approx 30$, where we let $\beta=0.1$, $K=8$ and $p=1/K$. It is observed that the MLA method converges to the DLMC method, when $n=20$. , When $n=10$, the converged result of the MLA method slightly deviates from the baseline. Note that we always used $M_1=10^3$, which effectively removed the statistical variance of the DLMC method.  
It took around $390$ s on a MacBook Air to obtain the MLA result for $n=20$ and $M_5=10^3$. In comparison, it took around $2000$ s to obtain comparable result of DLMC for $M_1=10^3$ and $N_1=10^5$. 

The MNIS results on the right of Figure \ref{fig:convergence} were produced using $20$ independent optimization runs. The horizontal axis represents $N_1$ and $N_6$ of the DLMC and MNIS methods, respectively. With few number of samples, the MNIS method is able to eliminate the bias of the approximation, hence, significantly outperforms the DLMC method. The CPU time of the MNIS method with $M_6=10^3$ and $N_6=10^3$, is very close to that of the MLA method with $M_5=10^3$, which produces similar results. 

\begin{figure}[ht!]
\centering
\includegraphics[scale=0.45]{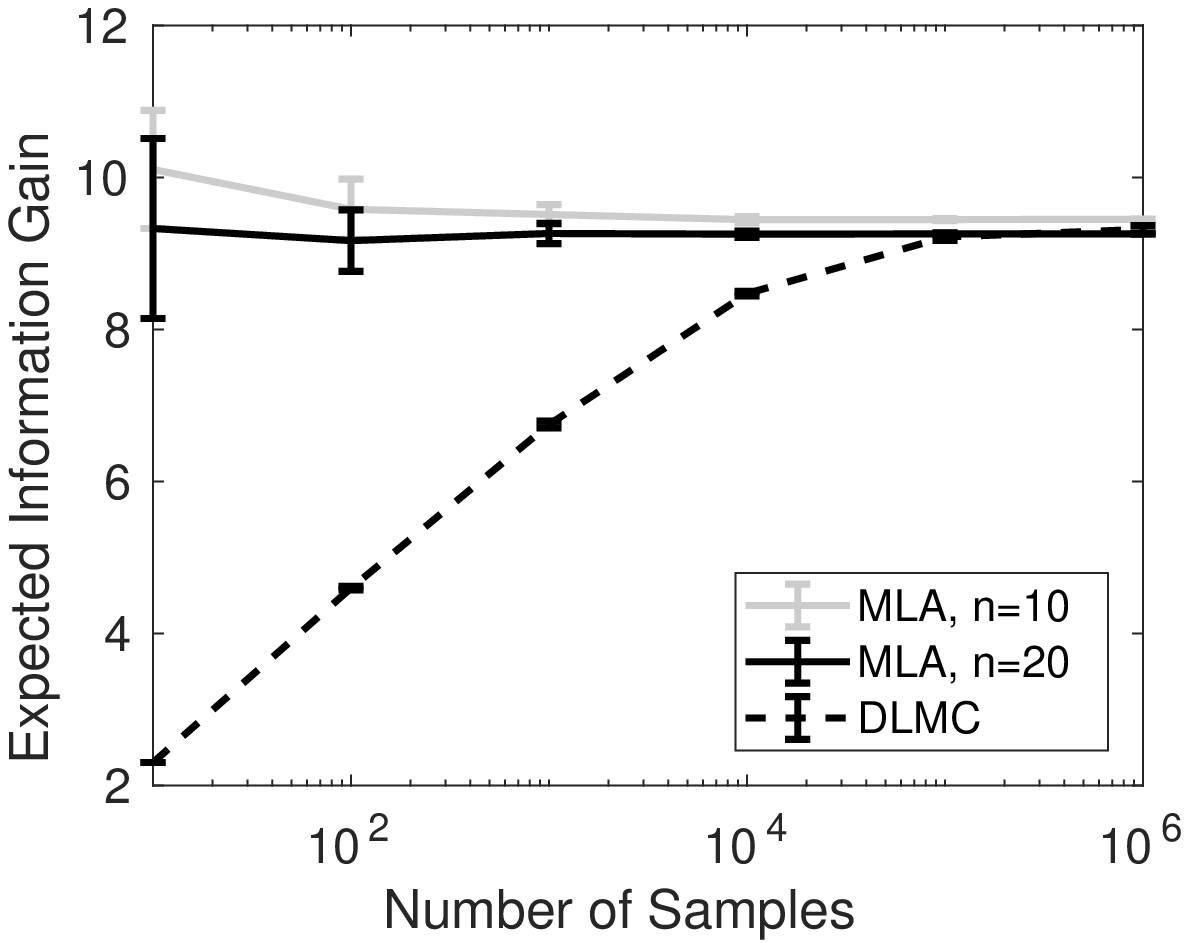}
\includegraphics[scale=0.4]{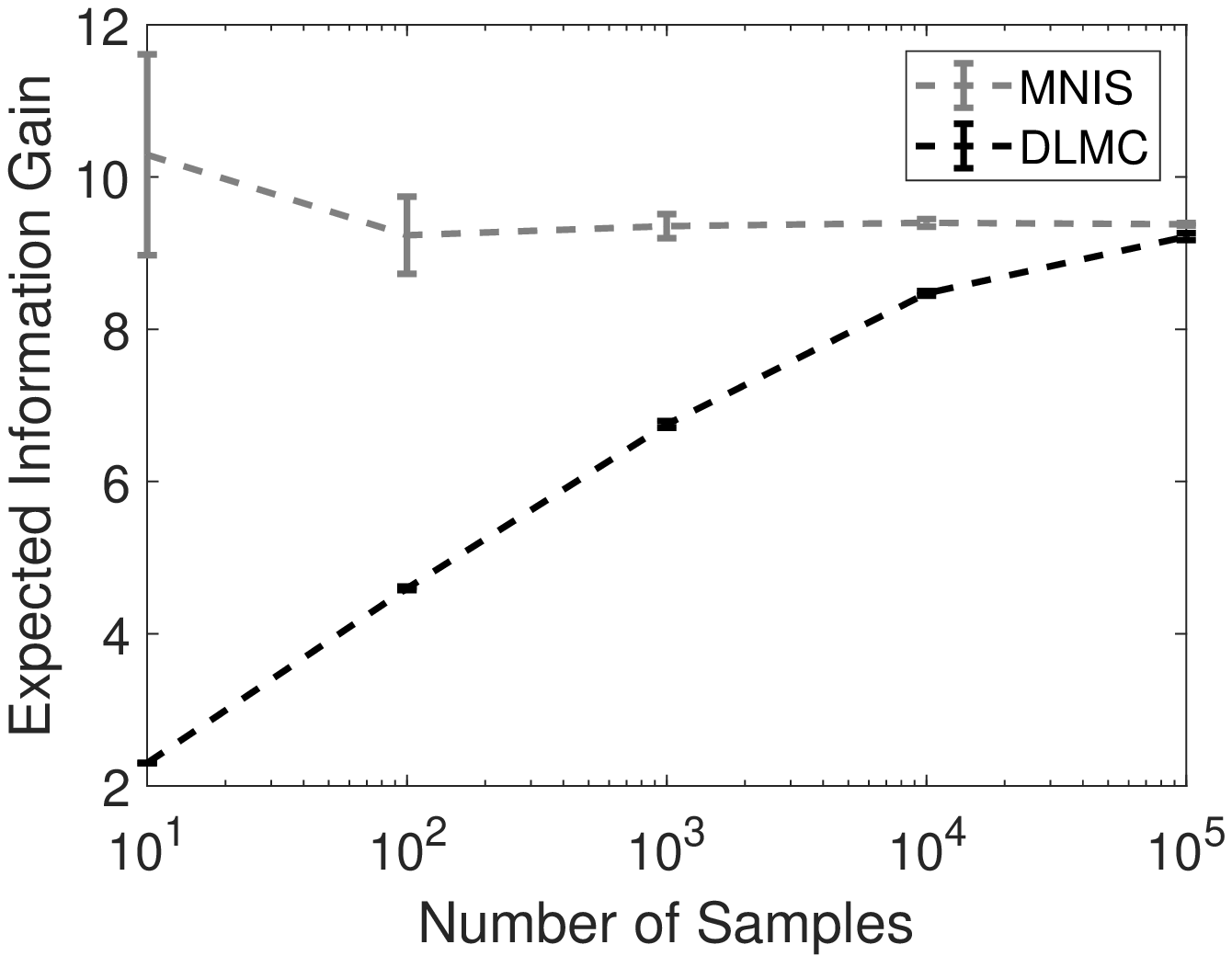}
\caption{The left figure shows the convergences of the estimated expected information gain using the MLA method and the DLMC method. $\xi= 1$. The right figure shows the convergence of the MNIS method and the DLMC method. $\xi= 1$.}\label{fig:convergence}
\end{figure} 

The initial starting point in each run of optimization is obtained according to the Latin hypercube sampling \cite{stein1987large}. 
Figure \ref{fig:I_inner_no_opt} demonstrates the effect of the number of optimization runs on the expected information gain and the averaged number of modes. The expected information gains were computed using the MLA method with $M_5=10^4$. It is observed that the relative error is around $2\%$, when  $10$ optimization runs are carried out. The relative error is smaller than $0.5\%$, when we used more than $20$ optimization runs. 

\begin{figure}[ht]
\centering
\includegraphics[scale=0.45]{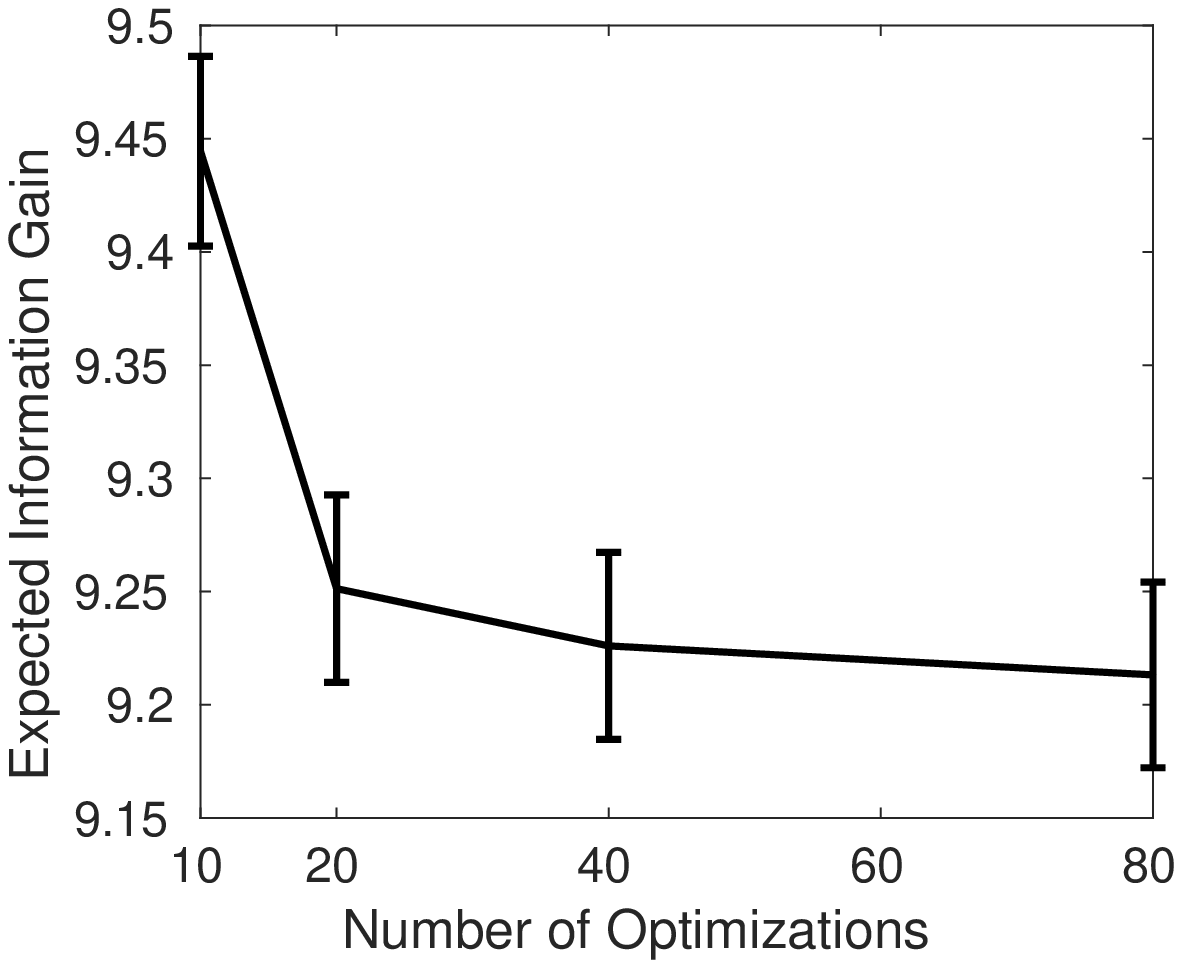}
\includegraphics[scale=0.45]{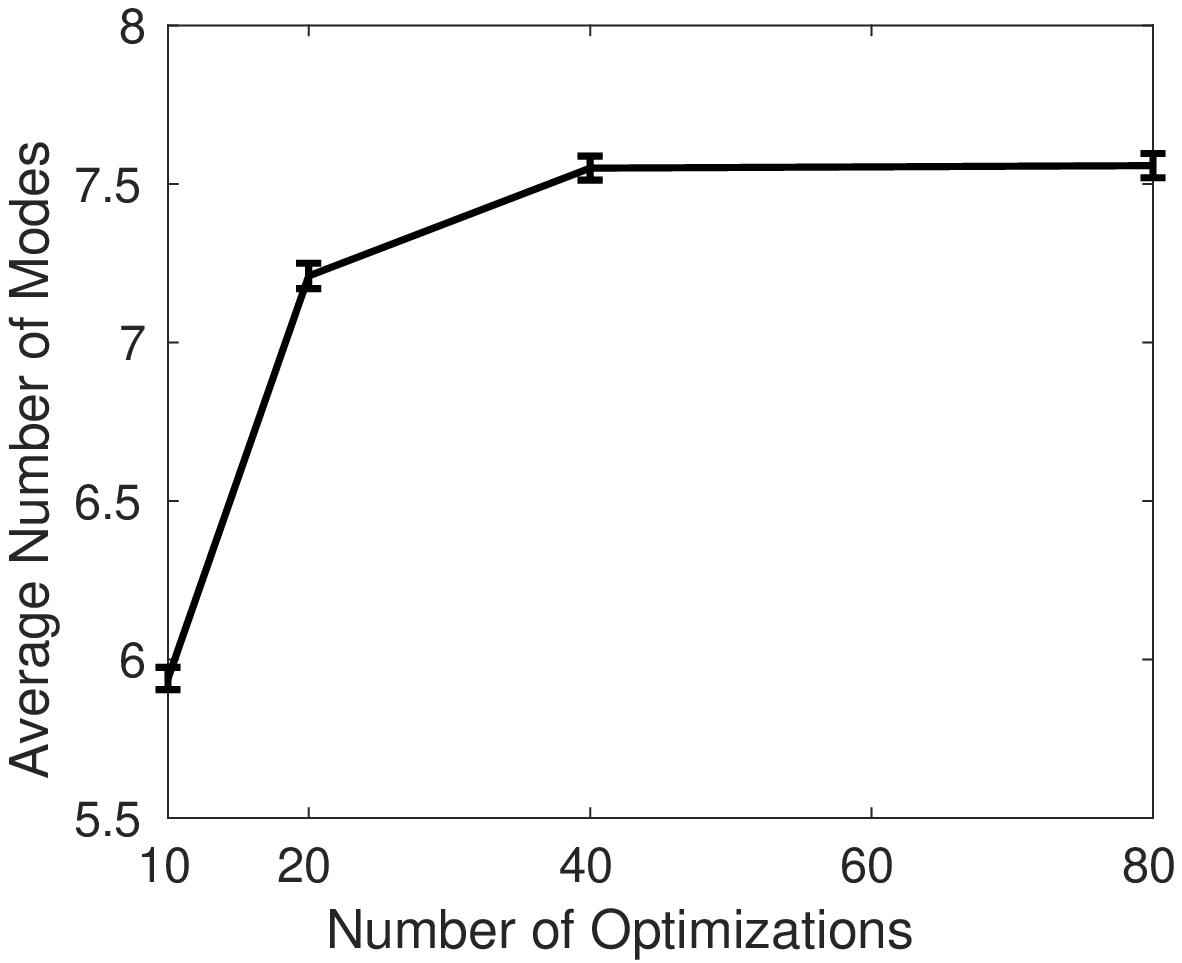}
\caption{The left figure shows the estimated expected information gains w.r.t. various numbers of optimization runs. $\xi= 1$. The right figure shows the average number of obtained modes w.r.t. various numbers of optimization runs. $\xi= 1$. } 
\label{fig:I_inner_no_opt}
\end{figure}



We plot the expected information gains of various experimental setup (i.e., $\xi \in [0.05, 0.55]$) in Figure \ref{fig:mnis-mla-dlmc1}. The variance of measurement noise is $\sigma_e^2 = 4$ on the left of Figure \ref{fig:mnis-mla-dlmc1}. It is clear that the MLA method is biased in this case due to the error term in \eqref{eq:theorem1}, while the MNIS method can exactly recover the baseline. 

Additionally, we consider a scenario, where the observation noise is not calibrated. Specifically, $\sigma_e$ is a random variable. By a little abuse of notation, the posterior pdf, $p(\vec\theta|\bar{\vec y}) = \int_{\sigma_e} p(\vec\theta|\bar{\vec y},\sigma_e) p(\sigma_e)d\sigma_e$, is now a marginal pdf against the unknown standard deviation of the Gaussian measurement noise. It can be estimated using $p(\vec\theta|\bar{\vec y}) = \frac{1}{L}\sum_{l=1}^L p(\vec\theta|\bar{\vec y},\sigma_e^l) \, \text{with} \, \sigma_e^l \sim p(\sigma_e)$. With few modifications to Algorithm 2, the expected information gain can then be estimated using Algorithm 3. 

\begin{algorithm}[pht]
\caption{Importance sampling based on multimodal Laplace approximation for uncalibrated observation noises}\label{alg:MMIS}
\begin{algorithmic} [1]
\STATE inputs: $p(\vec\theta)$, $g(\vec\theta)$, $\vec \xi$, $n$ 
\STATE draw two samples from the prior distributions 
$\vec\theta^i_t \sim p(\vec\theta)$ and $\sigma_e^i \sim p(\sigma_e)$  
\STATE draw a sample of data from the likelihood function $\bar{\vec y}_i \sim p(\bar{\vec y} | \vec\theta^i_t, \sigma^i_e)$  
\STATE solve problem $\arg\min_{\vec\theta} [-log(\frac{1}{L}\sum_{l=1}^L p(\vec\theta|\bar{\vec y}_i,\sigma_e^l))]$ $n$ times with distinct initial points, the distinct local optimal solutions are $\{\hat{\vec\theta}^i_k, k=1,...,K\}$ (note that $K$ can also vary w.r.t. $i$)
\STATE  draw $N_6$ i.i.d. samples from Gaussian mixture: $\vec\theta_j \sim \sum^K_{k=1} w_k {\cal{N}}( \hat{\vec\theta}^i_k, \vec\Sigma^i_k )$
\STATE  compute the inner sample average in \eqref{eq:MNIS}: $L_i=\frac{1}{N_6} \sum^{N_6}_{j=1} p(\bar{\vec y}_i | \vec \theta_j)\beta^i_j$ 
\STATE Repeat steps 2-6 $M_6$ times and compute the outer sample average in \eqref{eq:MNIS}:
$I_{MNIS}= \frac{1}{M_6}\sum^{M_6}_{i=1}log\left[\frac{p(\bar{\vec y}_i | \vec \theta_i)}{ L_i}\right]$ 
\end{algorithmic}
\end{algorithm}

We assume $\sigma_e \sim U(2,4)$ and compute the expected information gain of the uncalibrated experiments using Algorithm 3. The results are plotted on the right of Figure \ref{fig:mnis-mla-dlmc1}. The uncalibrated experiments provide less expected information gains than the calibrated experiments do.  
\begin{figure}[ht]
\centering
\includegraphics[scale=0.4]{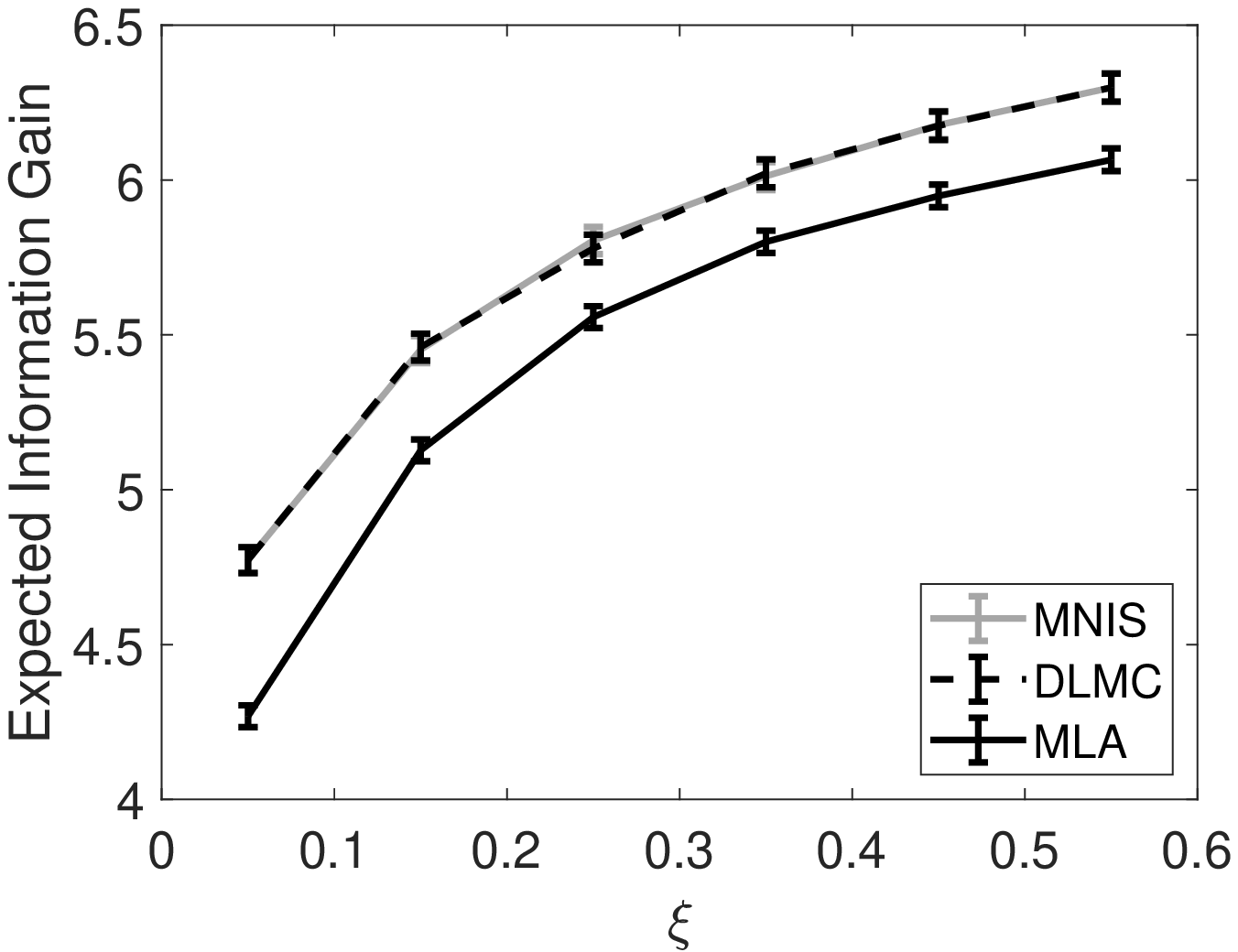}
\includegraphics[scale=0.3]{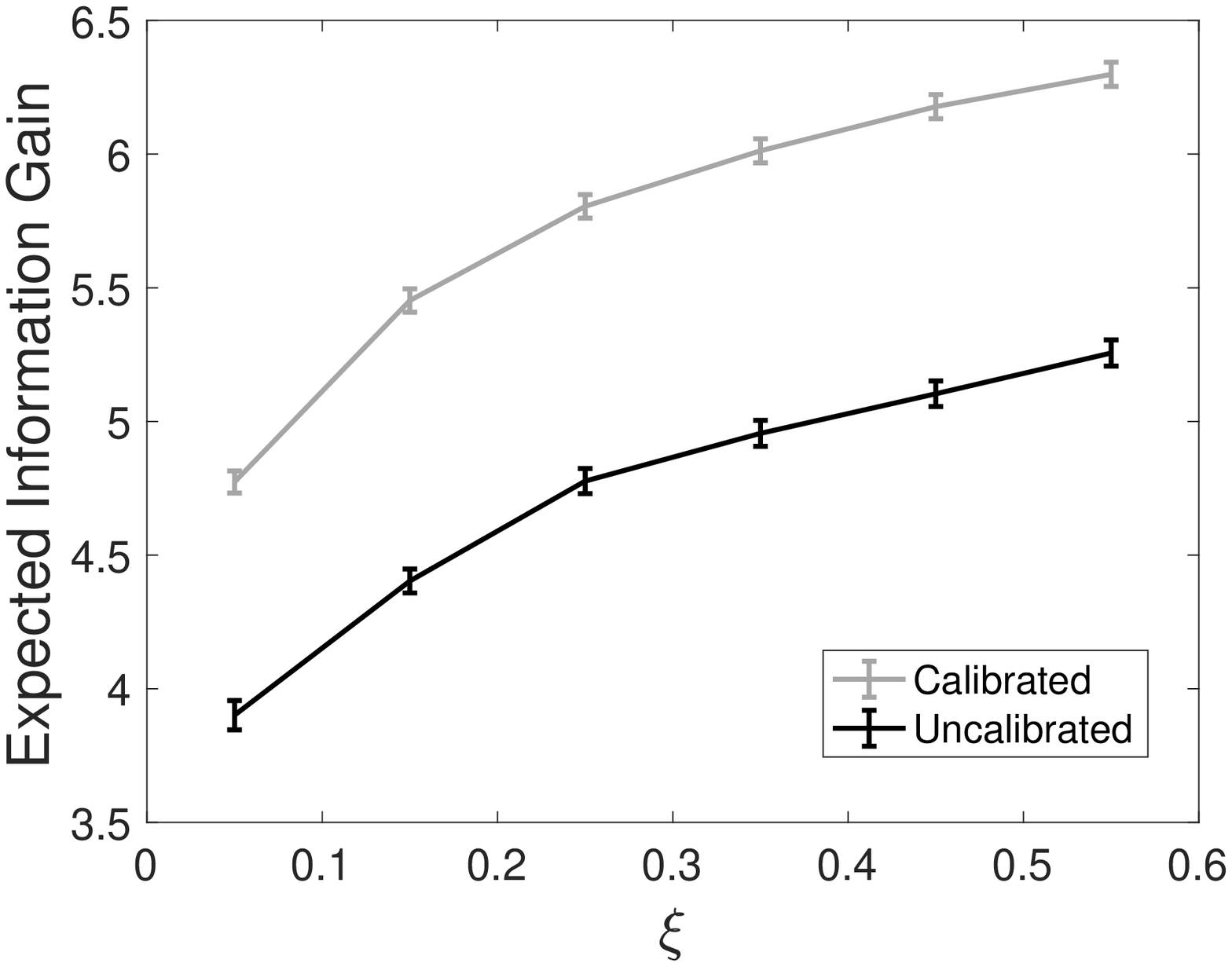}
\caption{The left figure demonstrates the expected information gains of various designs. $10^4$ and $10^5$ samples are used in the outer and inner loops of the DLMC estimator. $10^4$ samples are used in both the outer loop and the inner loop of the MNIS method. $10^4$ samples are used in the MLA method. The right figure demonstrates the expected information gains, when the observation noises are uncalibrated.}
\label{fig:mnis-mla-dlmc1}
\end{figure}


\subsection{Identification of sensor networks}

We infer the locations of the sensors in a two-dimensional space from some measurements of the distances between them. The posterior distributions of the locations are known to have multiple modes \cite{Tak2018}\cite{Ihler2005}\cite{Lan2014}. We demonstrate the proposed methodologies using two scenarios, where there are four and six sensors, respectively. 

In the first case, the coordinates ($x_1$, $z_1$, $x_2$ and $z_2$) of the first and the second sensors are unknown random variables distributed uniformly between zero and one, (i.e., $x_i \sim {\cal{U}}(0,1)$ and $z_i \sim {\cal{U}}(0,1)$, $i=1, 2$). The coordinates of the third and forth sensors are fixed at $x_3=0.5$, $z_3=0.3$, $x_4=0.3$ and $z_4=0.5$, respectively. Figure \ref{fig:sensors} shows a set of realizations of the four sensors on the left. 

\begin{figure}[ht]
\centering
\includegraphics[scale=0.4]{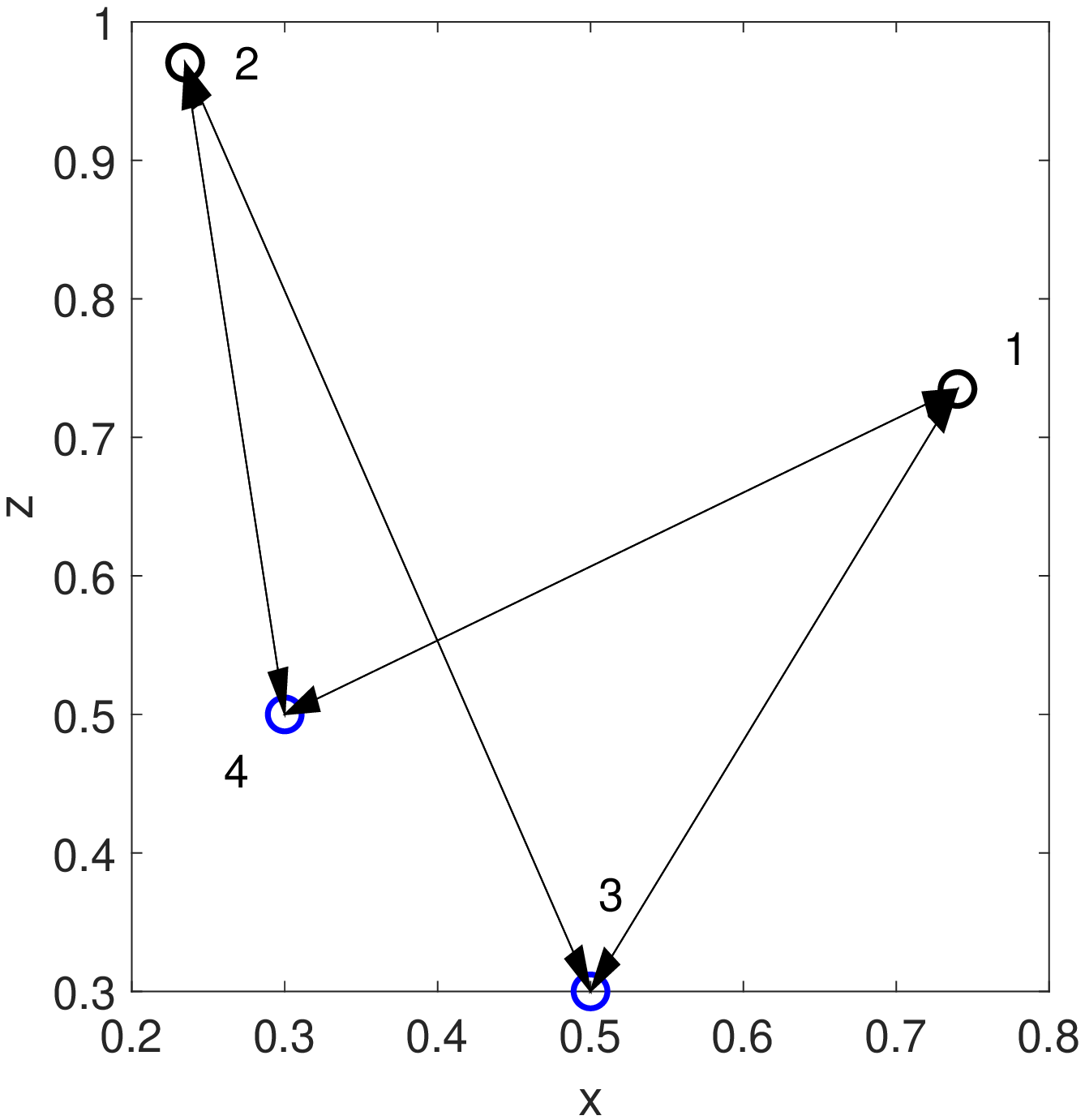}
\includegraphics[scale=0.4]{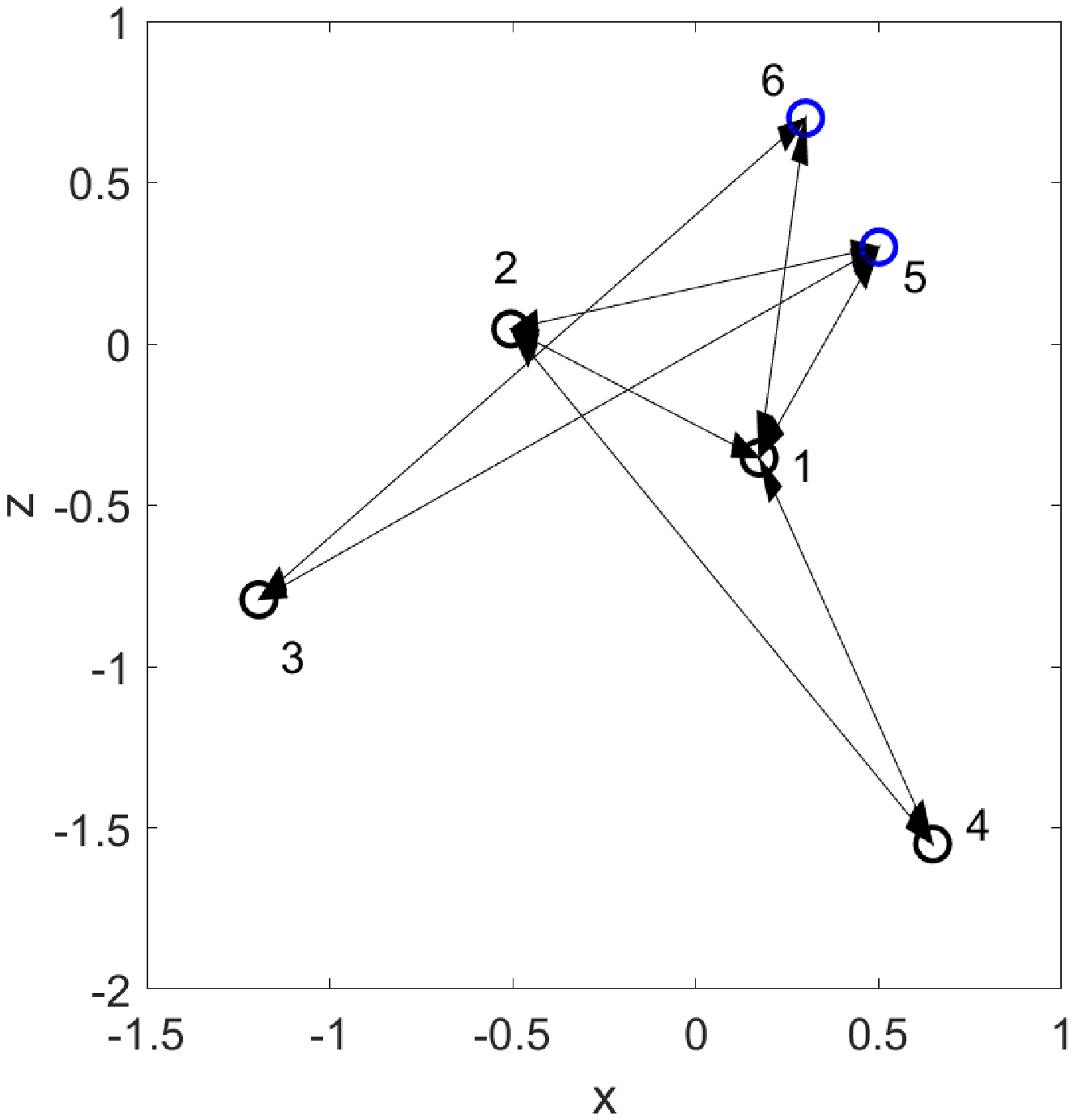}
\caption{The left figure shows a network of four sensors. The blue circles denote the fixed sensors. The black circles denote the sensors of unknown locations. The arrows connecting sensors indicate measurements of distances. The right figure shows a network of six sensors. }
\label{fig:sensors}
\end{figure}

The data model is \[ \vec y = \vec M \vec D + \vec\epsilon\, ,\]
where $\vec D$ is a six-dimensional vector:
\[\vec D = [D_{12},\, D_{13},\, D_{14},\, D_{23},\, D_{24},\, D_{34}]^{\top}\,,\] 
with $D_{ij}=\sqrt{(x_i -x_j)^2 + (z_i - z_j)^2}$. $\vec M$ is a matrix of $d_m$ rows and six columns. It is filled by $0$s and $1$s. We measure the distances between the first and third sensors, the first and fourth sensors, the second and third sensors, and the second and fourth sensors. Accordingly, the matrix $\vec M$ reads
\begin{center}
$\vec M =
\begin{bmatrix}
0 & 1 & 0 & 0 & 0 & 0\\
0 & 0 & 1 & 0 & 0 & 0\\
0 & 0 & 0 & 1 & 0 & 0\\
0 & 0 & 0 & 0 & 1 & 0
\end{bmatrix} 
\,.$
\end{center}

The measurement noises are independent multivariate Gaussian: $\vec \epsilon \sim {\cal{N}}(\vec 0, \vec \Sigma_e), \vec\Sigma_e = \vec I \sigma^2_e, \sigma^2_e=0.0005^2$. We then estimate the expected information gain using the fore-mentioned methods, namely, the MNIS, DLMC and MLA methods. The convergences w.r.t. the number of samples are plotted in Figure \ref{fig:mnis-mla-dlmc}. We used $10^3$ samples in the outer loop of the DLMC method and the MNIS method. The DLMC method takes ten million samples in the inner loop to reach a negligible bias. To the contrary, the MNIS and MLA methods achieved very accurate results, when few samples (around $10^2$) are used. We used $20$ optimization runs to search for the modes in the MNIS and MLA methods. It took around $74.28$ s on a MacBook Air to obtain the MLA result for $n=20$ and $M_5=10^2$. In comparison, it took around $2.25\times 10^6$ s to obtain comparable result of DLMC for $M_1=10^3$ and $N_1=10^7$. Note that all the methods can be equally accelerated by parallel computing.
Note that the horizontal axis of Figure \ref{fig:mnis-mla-dlmc} means differently for different methods (i.e., $N_1$ of the DLMC method, $N_6$ of the MNIS methods, and $M_5$ of the MLA method). It is a fair comparison in that we used identical number of samples in the outer loop of both the DLMC and MNIS methods (i.e., $M_1 = M_6$).    

\begin{figure}[ht]
\centering
\includegraphics[scale=0.4]{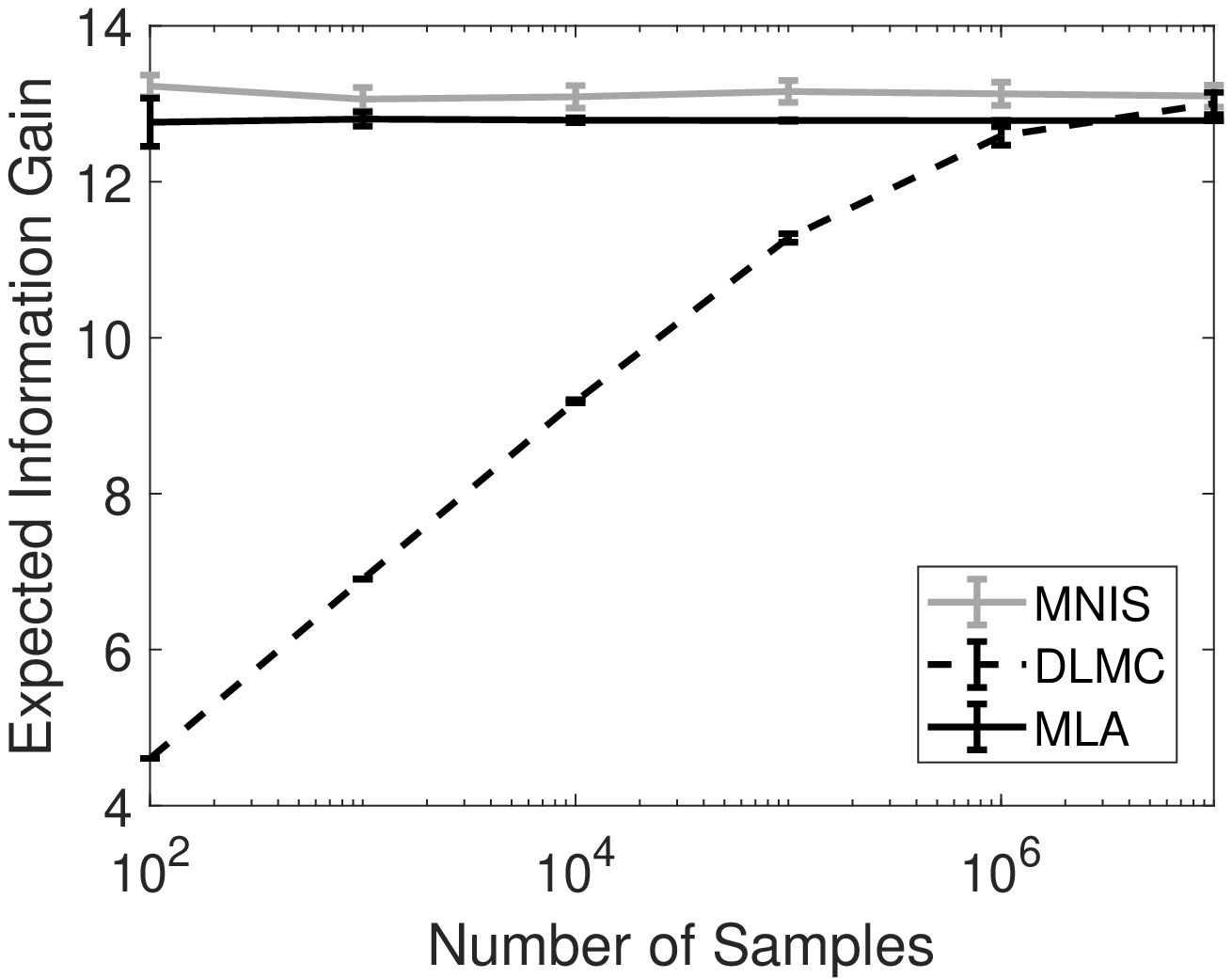}
\includegraphics[scale=0.4]{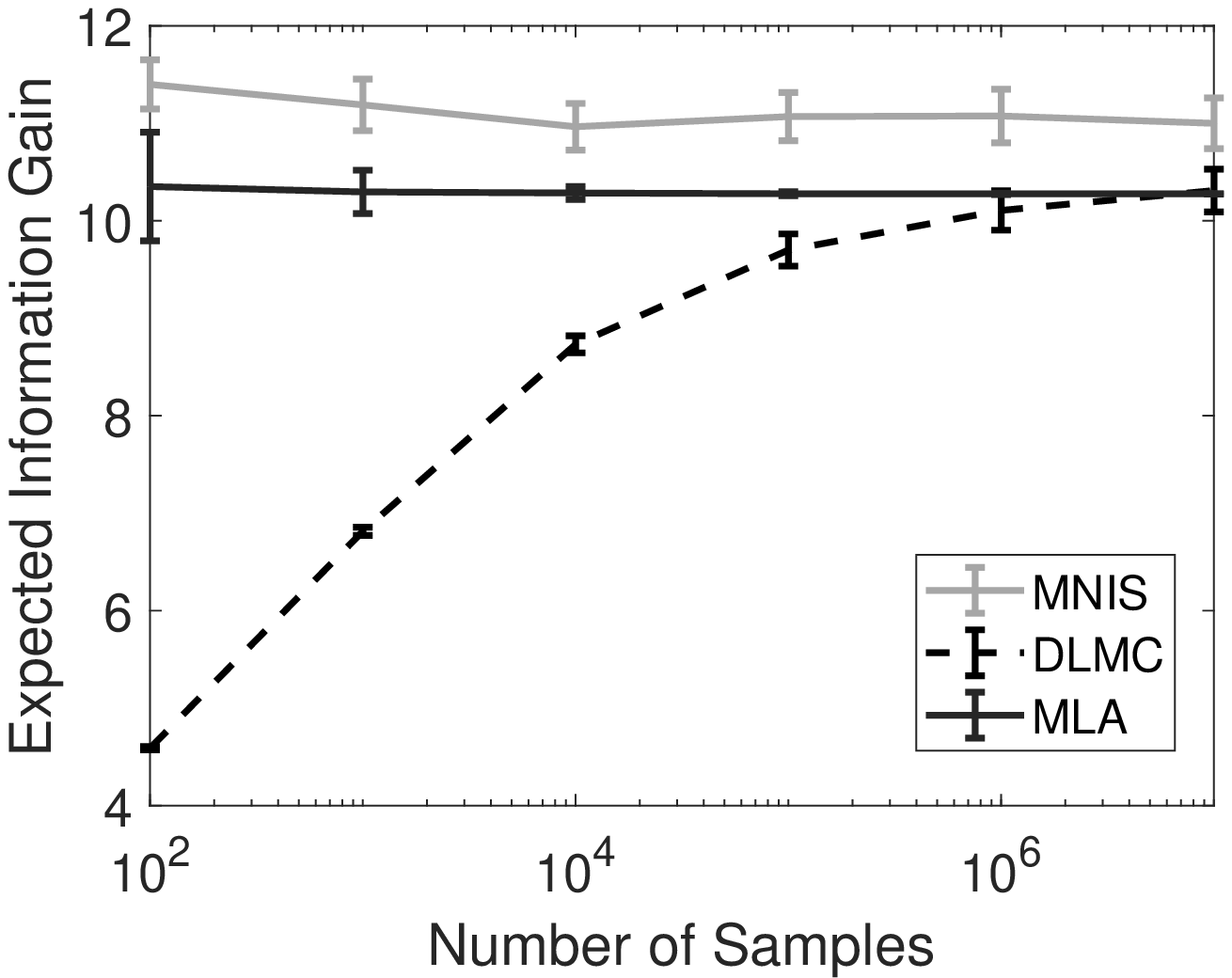}
\caption{The convergences of the expected information gain estimated using the MNIS , DLMC and MLA methods, respectively. $\sigma_e^2=0.0005^2$ in the four-sensor scenario on the left. $\sigma^2_e = 0.16$ in the six-sensor scenario on the right. }
\label{fig:mnis-mla-dlmc}
\end{figure}

In the second scenario shown on the right of Figure \ref{fig:sensors}, we increased the number of sensors to six, out of which the first four sensors are independently distributed according to a multivariate Gaussian distribution (i.e., $x_i \sim {\cal{N}}(0,1)$, and $z_i \sim {\cal{N}}(0,1)$, $i = 1,..., 4$). The locations of the fifth and sixth sensors are fixed at $x_5 = 0.5$, $z_5 = 0.3$, $x_6=0.3$ and $z_6=0.7$. $\vec D$ is a vector of $14$ components, which reads:
\[\vec D = [D_{12}, \, D_{13},\, D_{14},\, D_{15},\, D_{16},\, D_{23},\, D_{24},\, D_{25},\, D_{26},\, D_{34},\, D_{35},\, D_{36},\, D_{45},\, D_{46}]^{\top}\,, \]
where $D_{ij}=\sqrt{(x_i - x_j)^2 + (z_i - z_j)^2}$. The matrix $\vec M$ is the following:
\begin{center}
$
\vec M =
\begin{bmatrix}
1 & 0 & 0 & 0 & 0 & 0 & 0 & 0 &0 & 0 & 0 & 0\\
0 & 0 & 1 & 0 & 0 & 0 & 0 & 0 &0 & 0 & 0 & 0\\
0 & 0 & 0 & 1 & 0 & 0 & 0 & 0 &0 & 0 & 0 & 0\\
0 & 0 & 0 & 0 & 1 & 0 & 0 & 0 &0 & 0 & 0 & 0\\
0 & 0 & 0 & 0 & 0 & 0 & 1 & 0 & 0 &0 & 0 & 0\\
0 & 0 & 0 & 0 & 0 & 0 & 0 & 1 & 0 &0 & 0 & 0\\
0 & 0 & 0 & 0 & 0 & 0 & 0 & 0 & 0 &0 & 1 & 0\\
0 & 0 & 0 & 0 & 0 & 0 & 0 & 0 & 0 &0 & 0 & 1
\end{bmatrix} 
\,.$
\end{center}


The variance of measurement noise is $\sigma_e^2 =0.16^2$. We compare the convergences of the MNIS, DLMC and MLA methods on the right of 
Figure \ref{fig:mnis-mla-dlmc}. Again, we fixed the number of samples in the outer loop to be $1000$ for the DLMC and MNIS methods. It is observed that the DLMC method has a relative error of $10\%$ even with $N_1=10^7$. Then, we set $\sigma_e^2 =0.005^2$. The MNIS and the MLA methods can achieve accurate estimations with high confidence using just $100$ samples, while the DLMC method is erroneous as shown on the left of Figure \ref{fig:mnis-mla-dlmc_sensors_0005}. 


\begin{figure}[ht]
\centering
\includegraphics[scale=0.4]{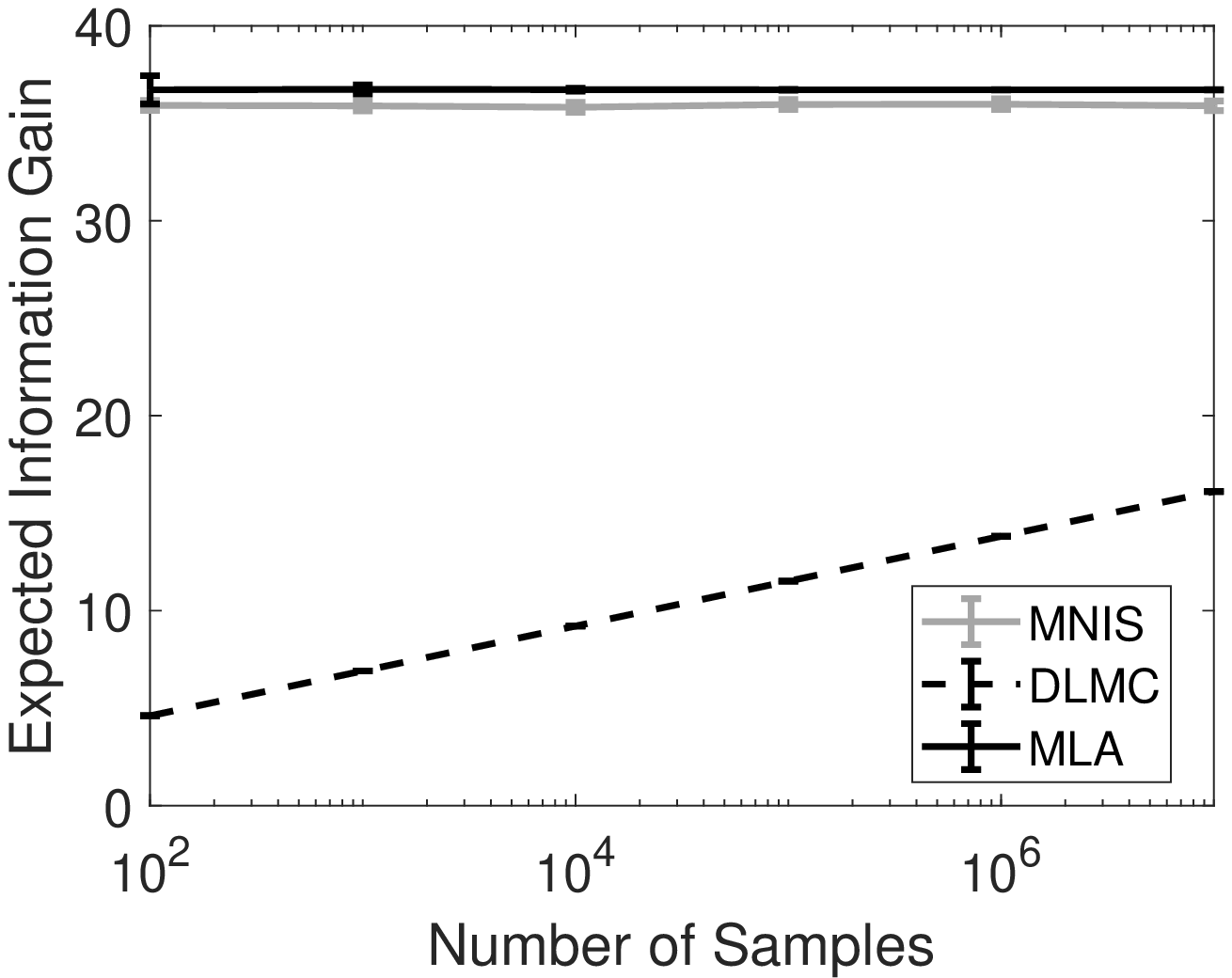}
\includegraphics[scale=0.4]{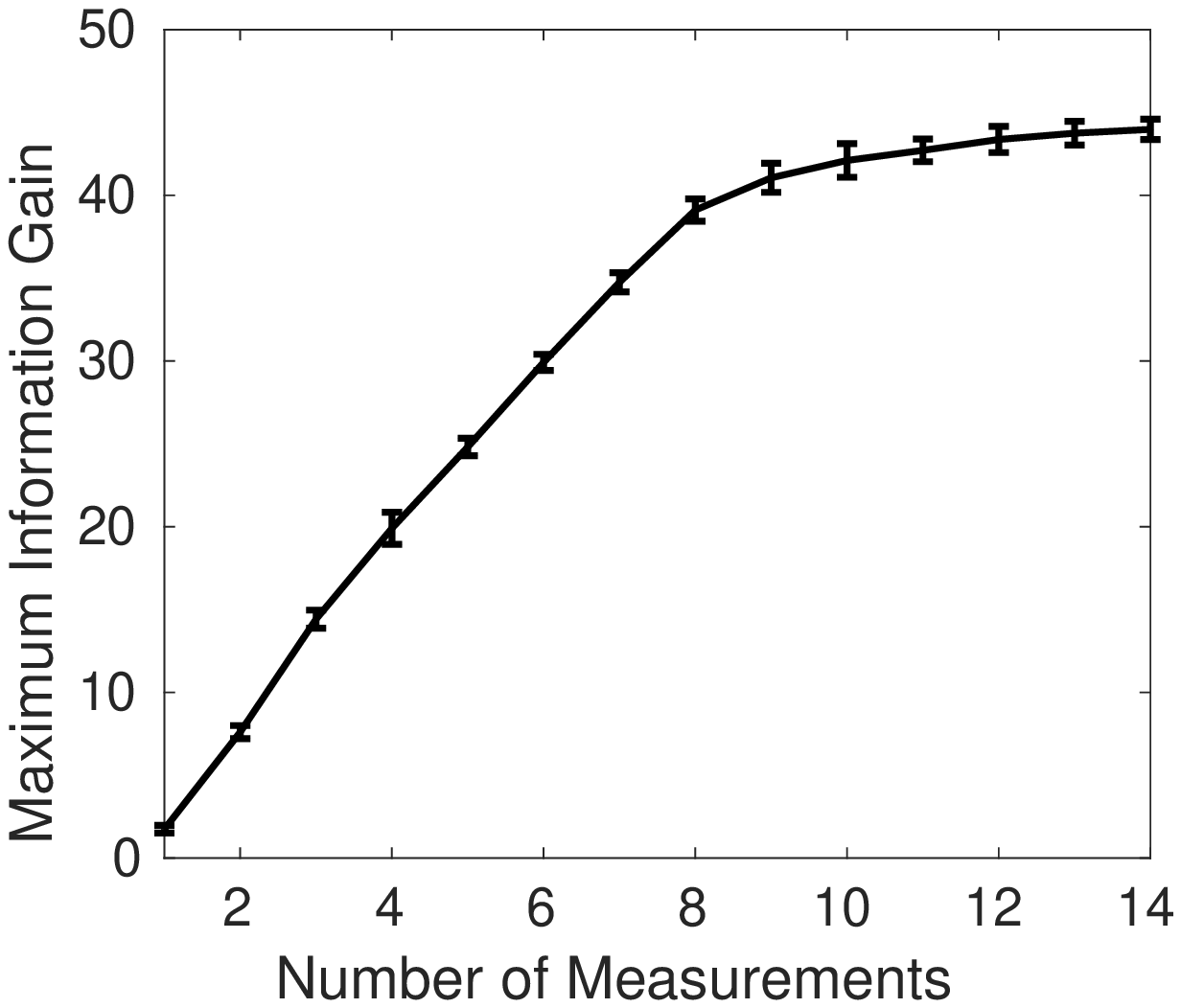}
\caption{The left figure shows the expected information gains computed using MNIS, DLMC and MLA methods, respectively. $\sigma^2_e = 0.005$. The right figure shows the maximum expected information gain of various number of measurements.}
\label{fig:mnis-mla-dlmc_sensors_0005}
\end{figure}

Figure \ref{fig:mnis-mla-dlmc_sensors_0005} on the right shows the maximum expected information gain w.r.t. various numbers of measurements (i.e., $d_m=1,...,14$). It is noted that the rate of information gain saliently slows down after eight of the fourteen unknown distances have been measured. Some of the best configurations for various numbers of measurements are plotted in Figure \ref{fig:bestDesigns}.


\begin{figure}[ht]
\centering
\includegraphics[angle=-90, scale=0.6]{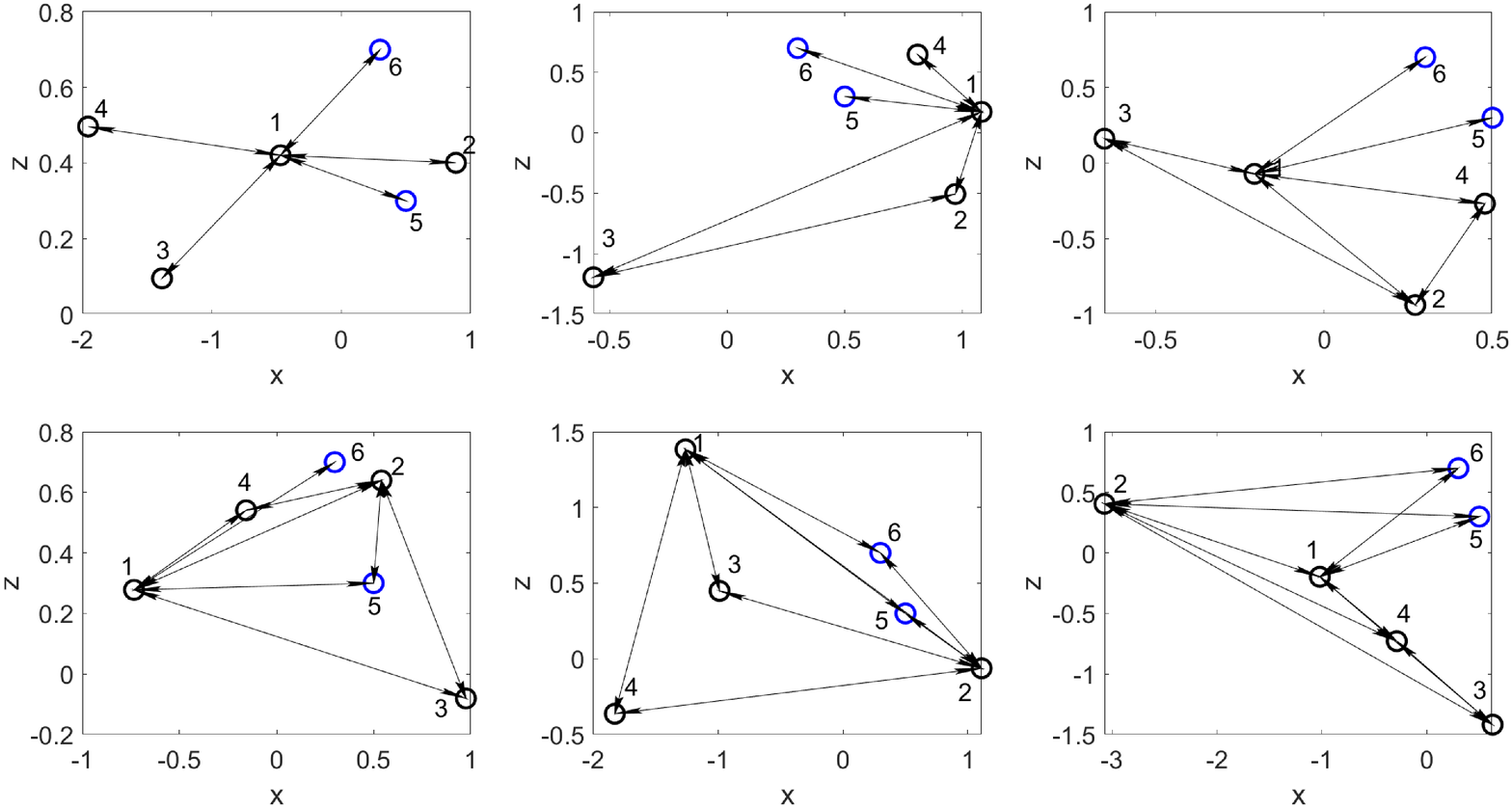}
\caption{The best designs w.r.t. various numbers of measurements.}
\label{fig:bestDesigns}
\end{figure}

\section{Conclusion}\label{sec:CC}
Two novel multimodal approaches have been devleoped to estimate the expected information gain in Bayesian optimal design of experiments. The methods are developed to consider both calibrated and uncalibrated observation noises. A multivariate Gaussian mixture distribution is constructed by few direct optimization runs with randomized initial guesses and weighted local Laplace approximations. In the first approach, the Gaussian mixture distribution is used to approximate the posterior pdf. Hence, the K-L divergence can be computed analytically. And, we can efficiently remove the nested integral in the DLMC method. In the second approach, the Gaussian mixture is used as the proposal pdf in an importance sampling scheme, hence, we can remove the possible Laplace bias in the first approach. It is shown both theoretically and numerically that the methods have similar computational costs. They can be several magnitudes more efficient than the DLMC methods. Additionally, the proposed methodologies are used to efficiently control the error of the expected information gain, therefore, the error in the numerical solution of the optimal design is also restrained. We consider applying the novel multimodal experimental design approach to models based on complex numerical PDEs as one future research direction.

\begin{acknowledgements}
The author would like to thank the associate editor and anonymous referees for their careful reading and constructive comments.
\end{acknowledgements}

%
%


\end{document}